%% file: main.tex
\documentclass[11pt]{article}
\usepackage{authblk}
\usepackage{appendix}
\usepackage[a4paper,margin=3cm]{geometry}
\usepackage[utf8]{inputenc}
\usepackage{amsmath}
\usepackage{amssymb}
\usepackage{xcolor}
\usepackage{amsthm}
\usepackage{dsfont}
\usepackage{color}
\usepackage{amsmath}
\usepackage{xcolor}

\usepackage{amsfonts}
\usepackage{amssymb}
\usepackage{bbm}
\usepackage{lineno}
\usepackage{multirow}
\usepackage{relsize}
\usepackage{enumitem}
\usepackage{pst-node}
\usepackage[compact]{titlesec}
\usepackage{tikz}
\usetikzlibrary{topaths,calc}
\usetikzlibrary{positioning}
\usetikzlibrary {shapes,matrix}
\usetikzlibrary{arrows}
\usepackage{changepage}
\tikzset{
	semithick,
	node distance = 2cm,
	dot/.style={circle,fill,inner sep=2pt}
}
\tikzset{
	side by side/.style 2 args={
		line width=2pt,
		#1,
		postaction={
			clip,postaction={draw,#2}
		}
	}
}
\tikzstyle{every state}=[draw = black,thick,fill = white,minimum size = 4mm]
\tikzstyle{selected edge} = [draw,line width=2pt,-,red!50]
\tikzset{
	edge/.style={->,> = latex'}
}
\usepackage[colorlinks=true,linkcolor=red,citecolor=blue]{hyperref}
\usepackage[nameinlink]{cleveref}

\newcommand{\comment}[1]{}
\usepackage[colorlinks=true,linkcolor=red,citecolor=blue]{hyperref}
\usepackage[section]{placeins}
\usepackage[nameinlink]{cleveref}
\usepackage[procnumbered, linesnumbered,algoruled]{algorithm2e}

\newcommand{\R}{{\mathbb{R}}}
\renewcommand{\P}{{\mathbb{P}}}

\newcommand{\OPT}{\textnormal{OPT}}

\newcommand{\E}{{\mathbb{E}}}

	\Crefname{obs}{Observation}{Observations}
		\Crefname{cor}{Corollary}{Corollaries}

\begin{document}
\newtheorem{thm}{Theorem}[section]
\newtheorem{prop}[thm]{Proposition}
\newtheorem{assm}[thm]{Assumption}
\newtheorem{lem}[thm]{Lemma}
\newtheorem{obs}[thm]{Observation}
\newtheorem{cor}[thm]{Corollary}
 \newtheorem{lemma}[thm]{Lemma}
  \newtheorem{definition}[thm]{Definition}
 \newtheorem{theorem}[thm]{Theorem}
 \newtheorem{proposition}[thm]{Proposition}
 \newtheorem{claim}[thm]{Claim}
\newtheorem{defn}[thm]{Definition}
\newcommand{\ariel}[1]{{\color{red} (Ariel :#1)}}
\def \II   {{\mathcal I}}
\newcommand{\one}{\mathbbm{1}}
	\def\claimproof{\proof}
\def\endclaimproof{\unskip\nobreak\hfill$\square$\par\vspace{0pt}}
\SetKwIF{If}{ElseIf}{Else}{if}{then}{else if}{else}{}      
\SetKwFor{For}{for}{do}{}
\date{}
\title{ {\bf Incremental Dominating Set}}

\author[2]{Ilan Doron Arad\thanks{\texttt{ilanda@mit.edu}}}
\author[1]{Jonathan Gal\thanks{\texttt{jonathan.gal@campus.technion.ac.il}}}
\author[1]{Seffi Naor\thanks{\texttt{naor@cs.technion.ac.il}}}

\affil[1]{Technion -- Israel Institute of Technology,
Haifa, Israel}
\affil[2]{Massachusetts Institute of Technology\\
Cambridge, MA, USA}
\maketitle
\begin{abstract}
Dominating Set is a fundamental problem in graph theory: given a graph, find a minimum-weight subset of vertices such that every vertex is either selected or adjacent to a selected vertex. In online settings where vertices arrive sequentially, comparing algorithms against an offline optimum with full knowledge of the input leads to extremely strong lower bounds, where even a simple star graph shows that any online algorithm must have competitive ratio $\Omega(\Delta)$, with $\Delta$ the largest degree of any vertex in the graph, matching the trivial strategy of selecting all vertices.

We study the \textit{incremental} dominating set problem, where the optimal algorithm is constrained to the same choices available to online algorithms. This introduces a benchmark that enables a meaningful comparison between algorithms.

We present the first results for vertex-weighted graphs and randomized algorithms in this model. For incremental dominating set, we give an $O(\Delta)$-competitive deterministic algorithm and an $O(\log^2\Delta)$-competitive randomized algorithm. We extend these results to the Connected Dominating Set problem using a linear-programming formulation that captures connectivity through local constraints. When the neighborhood of each arriving vertex is known \textit{in advance}, deterministic algorithms achieve similar polylogarithmic competitive ratios as their randomized counterparts. Finally, we establish matching lower bounds, showing that our results are optimal up to constant factors.
\end{abstract}
\section{Introduction}
Dominating Set is a fundamental and extensively studied problem in graph theory and algorithm design, e.g.~\cite{hedetniemi1991bibliography, konig12theorie, ore1962theory}.
Given a graph $G=(V,E)$ with vertex weights $w: V\to\R^+$, we say that vertex $u$ dominates vertex $v$ if either $u=v$ or $(u,v)\in E$.
A subset of vertices $D\subseteq V$ is then called a dominating set if every vertex in $V$ is dominated by at least one vertex in $D$. Variants include the Connected Dominating Set problem~\cite{cdsbook,guha1998approximation, yu2013connected}, which requires the dominating set to be connected, and the Total Dominating Set problem~\cite{atapour2008total, cockayne1980total}, where every vertex must be dominated by a distinct neighbor. These, and related problems, have found applications in network routing~\cite{wu2001dominating}, facility location~\cite{hooker1991finite,mehrez1984facility}, efficient transmission~\cite{liu1968introduction}, and more~\cite{milner1964theory}.

We consider online versions of these problems. In the online setting, domination requests arrive sequentially, requiring the algorithm to make irrevocable decisions, while maintaining feasibility at each step without knowledge of future requests~\cite{main, das1997routing, kobayashi2017improved}. 

Online domination has primarily been studied under two models: the known-graph and unknown-graph settings. In the known-graph model, the entire graph is known upfront, and requests for dominating specific vertices arrive sequentially.
This model is a special case of the online set cover problem, for which the known competitive bounds are polylogarithmic in the number of sets and elements~\cite{osc}. Several works have extended these bounds to variants of the known-graph model in the context of domination problems~\cite{markarian2020online}. 

In the unknown-graph model, the vertices are revealed one by one along with the edges connecting them to previously seen vertices. Each revealed vertex must be immediately dominated. The version most commonly studied in this model is the late accept variant~\cite{main}, in which a vertex can be added to the dominating set at any time after it arrives, but cannot be removed once it is added.

Most research in this setting analyzes online algorithms by comparing them to a powerful offline algorithm that can choose any subset of vertices, independently of their arrival order. Consider the star graph in Figure~\ref{fig:star-counterexample}; it demonstrates that every online algorithm must have competitive ratio $\Theta(\Delta)$, where $\Delta$ is the maximum degree of a vertex, matching the trivial strategy that always adds all vertices to the dominating set. Thus, there has been very little research in this direction ~\cite{main,de2021online}.

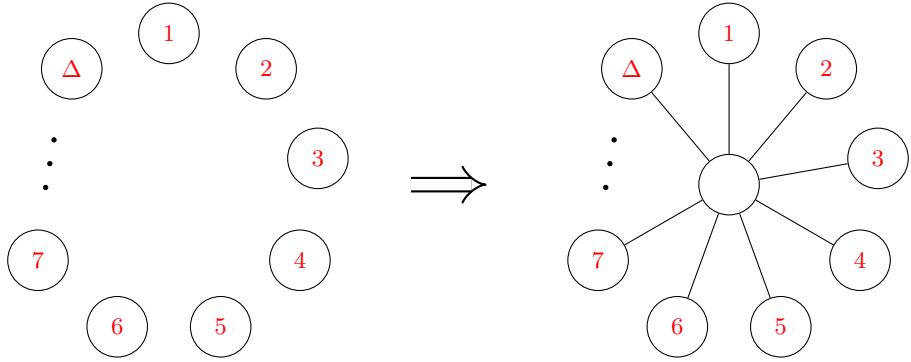
\begin{figure}[ht]
    \centering
    \begin{tikzpicture}
    \node[circle, draw, minimum size=8mm, font=\footnotesize] at (360:0mm) (center) {\textcolor{blue}{}};
    \node[coordinate] (D2)  [left=2.8cm of center] {};
    \node[coordinate] (centerwo)  [left=7cm of center] {};
	\node[coordinate] (C1)  [right=3.2cm of centerwo] {}; 
	\draw[thick, double,-implies, double distance=1.5mm] (C1)--(D2);
    \foreach \n in {1,...,7}{
        \node[circle, draw, minimum size=8mm, font=\footnotesize] at ({-(\n - 1)*360/9 + 90}:2cm) (n\n) {\textcolor{red}{$\n$}};
        \draw (center)--(n\n);
    }
    \foreach \i in {1,...,7} {
    \pgfmathsetmacro{\angle}{90 - (\i-1) * 40} 
    \path (centerwo) ++(\angle:2cm) coordinate (p\i);
    \node[circle, draw, minimum size=8mm, font=\footnotesize] at (p\i) (m\i) {\textcolor{red}{$\i$}};
  }
  \pgfmathsetmacro{\angle}{90 - (9-1) * 40}
    \path (centerwo) ++(\angle:2cm) coordinate (p9);
    \node[circle, draw, minimum size=8mm, font=\footnotesize] at (p9) (m9) {\textcolor{red}{$\Delta$}};
  \pgfmathsetmacro{\angle}{90 - (8-1) * 40}
    \path (centerwo) ++(\angle:1.6cm) coordinate (p8);

    \node[rotate= 80, font=\huge] at (p8) (n8) {\textbf{$\ldots$}};
    \node[rotate= 80, font=\huge] at ({-(8 - 1)*360/9 + 90}:1.6cm) (n8) {\textbf{$\ldots$}};
    \node[circle, draw, minimum size=8mm, font=\footnotesize] at ({-(9 - 1)*360/9 + 90}:2cm) (n9) {\textcolor{red}{$\Delta$}};
        \draw (center)--(n9);
\end{tikzpicture} 
    \caption{A star graph where the center vertex is revealed last. An incremental algorithm, unable to choose the center vertex before it is requested, must select all leaf vertices, while an offline algorithm needs only to select the center vertex, yielding a competitive ratio of $\Omega(\Delta)$. However, against the optimal \textbf{incremental} solution, which is also unable to choose the center vertex, the competitive ratio is $1$.}
    \label{fig:star-counterexample}
\end{figure}
To provide a more meaningful benchmark, we propose comparing online algorithms 
with their incremental counterparts. In the incremental dominating set (DS) 
problem, we are given a graph $G=(V,E)$ and an ordering $v_1,v_2,\dots,v_n$ 
of the vertices. The goal is to choose a subset $S\subseteq V$ such that for 
every $t\in [n]$, the set $S\cap \{v_1,\dots,v_t\}$ is a dominating set for 
the subgraph induced by $\{v_1,\dots,v_t\}$. Since both the algorithm and the 
optimal incremental solution face the same arrival constraints, this benchmark 
is more informative. In the example in Figure~\ref{fig:star-counterexample}, 
the competitive ratio is reduced to $1$ from $\Omega(\Delta)$.

This problem naturally models constraints that online algorithms face in the unknown-graph setting. Boyar et al.~\cite{main} were the first to consider deterministic online algorithms  for the incremental model on unweighted graphs in the unknown-graph setting. They showed that any deterministic online algorithm has a competitive ratio of $\Theta(\Delta)$. This result implies that, for unweighted graphs, no deterministic online algorithm can substantially outperform the na\"ive approach of simply selecting all vertices. However, the question of obtaining improved competitive bounds for other settings, such as weighted graphs, where achieving an $O(\Delta)$ competitive algorithm is non-trivial, or randomized algorithms, has remained open.

\subsection{Our Contribution \& Results}
We address several open questions in online domination problems: obtaining competitive ratios for weighted graphs, randomized algorithms, and settings with additional structural information.

We start by formulating both incremental Dominating Set (DS) and incremental Connected Dominating Set (CDS) as covering problems. For incremental CDS we manage to encode both domination and connectivity requirements using a small number of linear constraints. This is the first use of such formulation in the online setting.

In the deterministic weighted setting, we show that incremental DS and incremental CDS admit an $O(\Delta)$-competitive algorithm, matching the unweighted lower bound~\cite{main}. For randomized algorithms, we achieve $O(\log^2 \Delta)$ competitive ratio and prove this is optimal up to a constant factor for polynomial-time algorithms. For computationally unbounded algorithms we prove an $\Omega(\log \Delta)$ lower bound.

We also study the \textbf{known neighborhood} model, where requesting a 
vertex reveals its full neighborhood. Surprisingly, knowing just the neighborhood of requested vertices in advance yields substantially stronger {\em deterministic} bounds: $O(\log^2 \Delta_G)$-competitive in the unweighted case for both incremental DS and incremental CDS where $\Delta_G$ is the maximum degree of the revealed graph (containing vertices which are not requested), and $O(\log n \log \Delta)$-competitive for weighted incremental DS. These competitive ratios hold at every step of the algorithm, and in particular the final solution remains competitive even if some revealed vertices were never requested. These bounds are asymptotically tight.

We also study offline and random-order models. In the offline setting, we establish tight $\Theta(\log \Delta)$ polynomial time approximation bounds unless NP $\subseteq$ BPP. In the random-order setting, we prove an $\Omega(\log \Delta)$ lower bound, showing that random arrival order does not substantially improve the competitive ratio.

A summary of our upper and lower bounds across different models appears in Table~\ref{fig:results}. We now discuss our results in detail.
\begin{table}[t]
    \centering
    \renewcommand{\arraystretch}{1.5}
\resizebox{\textwidth}{!}{%
\begin{tabular}{|c|c|c|c|}
\hline
&\textbf{Incremental DS} & \textbf{Incremental CDS} & \textbf{Lower Bound}\\
\hline
    Deterministic &$O(\Delta)$&$O(\Delta)$&$\Omega(\Delta)$~\cite{main}
    \\
    \hline
    Randomized (Polynomial Time) &\multirow{2}{*}{\centering $O(\log^2\Delta)$}&\multirow{2}{*}{\centering $O(\log^2\Delta)$}&$\Omega(\log^2\Delta)$\\
    \cline{1-1} \cline{4-4}
    Randomized (General)&&&$\Omega(\log\Delta)$\\
    \hline
    
    Known Neighborhood Unweighted&$O(\log^2\Delta_G)$&$O(\log^2\Delta_G)$&$\Omega(\log^2\Delta)$\\
    \hline
    Known Neighborhood (Polynomial Time) &\multirow{2}{*}{\centering $O(\log n\log\Delta)$}&\multirow{2}{*}{\centering -}&$\Omega(\log^2\Delta)$\\
    \cline{1-1} \cline{4-4}
    Known Neighborhood (Deterministic)&&&$\Omega(\log^2\Delta/ \log \log \Delta)$\\
    \hline
    
    Random Order Arrival&-&-&$\Omega(\log \Delta)$\\
    \hline
    Offline (Known Input) &$O(\log \Delta)$&$O(\log \Delta)$&$\Omega(\log \Delta)$\\
    \hline
\end{tabular}
}
    \caption{Results for incremental DS and incremental CDS. 
Entries marked with - follow from results in other rows. All results apply to weighted graphs unless specified otherwise. $\Delta$ refers to the maximum degree of the graph induced by requested vertices, while $\Delta_G$ refers to the maximum degree of the revealed graph.}
    \label{fig:results}
\end{table}
\subsubsection{Online Incremental}
We begin with the online setting of incremental DS and CDS. Boyar et al.~\cite{main} have shown a $\Theta(\Delta)$ deterministic lower bound for incremental DS and incremental CDS,
which matches the competitive bound of the trivial algorithm that selects all vertices in unweighted graphs. In contrast, we design a deterministic algorithm with a non-trivial competitive ratio for weighted graphs.

\begin{theorem}
\label{thm1.3}
There is a $(\Delta + 1)$-competitive deterministic algorithm for \textnormal{incremental DS} and \textnormal{incremental CDS}.
\end{theorem}
In the randomized setting, the competitive ratio improves exponentially.

\begin{theorem}
\label{thm1.4}
There is an $O(\log^2\Delta)$-competitive randomized algorithm for \textnormal{incremental DS} and \textnormal{incremental CDS}.
\end{theorem}

We note that this result gives the first sublinear-competitive algorithm for weighted online CDS. In contrast, in the known-graph model (or in any model in which we only need to dominate a subset of the vertices) the problem reduces to non-metric TSP, for which obtaining an $f(n)$-approximation is hard for every polynomial-time computable function~\cite{guha1998approximation}. Thus, the incremental model allows substantially better guarantees.

We also study the \textbf{known neighborhood} model, where requesting a vertex $v_i$ reveals its full neighborhood $N[v_i]$. Importantly, this extra information does not change the original setting, as the neighbors of $v_i$ become visible to the algorithm, but cannot be added to the dominating set until they are requested. Surprisingly, this additional information enables deterministic algorithms to achieve polylogarithmic competitive ratios matching their randomized counterparts.

\begin{theorem}
\label{thm1.11}
When the neighborhood of requested vertices is known, and the graph is unweighted, then
there is an $O(\log^2\Delta_G)$-competitive deterministic algorithm for both \textnormal{incremental DS} and \textnormal{incremental CDS}.
\end{theorem}
Here $\Delta_G$ denotes the maximum degree of the fully revealed graph, which may contain vertices that are never requested. The competitive ratio holds at every step of the algorithm, and remains valid even if the adversary stops requesting vertices before all revealed vertices have been requested. In the case where all revealed vertices are eventually requested we have $\Delta = \Delta_G$ and the algorithm is $O(\log^2\Delta)$-competitive.

In the weighted case we still get polylogarithmic results for incremental DS.
\begin{theorem}
\label{thm1.12}
When the neighborhood of requested vertices is known,  there is an $O(\log n \log \Delta)$-competitive deterministic algorithm for \textnormal{incremental DS}.
\end{theorem}
Here $n$ denotes the number of revealed vertices and $\Delta$ the maximum degree of $G[V_n]$, the subgraph induced by requested vertices. As in the unweighted case, the competitive ratio holds at every step of the algorithm.

We establish matching lower bounds for our randomized algorithms.
\begin{theorem}
\label{thm1.7}
There is no polynomial-time online algorithm for \textnormal{incremental DS} or \textnormal{incremental CDS} with competitive ratio $o(\log^2 \Delta)$ unless \textnormal{NP} $\subseteq$ \textnormal{BPP}, even when the neighborhood of requested vertices is known.
\end{theorem}

For deterministic algorithms, we obtain a stronger lower bound.
\begin{theorem}\label{thmdetLB}
There is no deterministic online algorithm for \textnormal{incremental DS} or \textnormal{incremental CDS} with competitive ratio $o(\log^2 \Delta/\log \log \Delta)$, even when the neighborhood of requested vertices is known.
\end{theorem}

Finally, in the random-order setting where vertices arrive according to a uniformly random permutation, we establish a logarithmic lower bound.

\begin{theorem}
\label{thm1.10}
There is no online random-order algorithm for \textnormal{incremental DS} or \textnormal{incremental CDS} with competitive ratio $o(\log \Delta)$.
\end{theorem}

The random-order model is a restriction of the fully adversarial model, the above lower bound immediately implies the following.

\begin{cor}
\label{thm1.8}
There is no online algorithm for \textnormal{incremental DS} or \textnormal{incremental CDS} with competitive ratio $o(\log \Delta)$. This lower bound also holds when the neighborhood of requested vertices is known.
\end{cor}

\subsubsection{The Offline Setting}
\label{sec:ell}
We study the offline variants of incremental DS and CDS, where the entire graph and the vertex request order are known in advance. We show that these problems remain NP-complete and hard to approximate, and establish tight approximation bounds up to constant factors.
\begin{theorem}
\label{thm1.1}
    \textnormal{Incremental DS} and \textnormal{incremental CDS} are \textnormal{NP-complete}. Furthermore, unless \textnormal{NP} $\subseteq$ \textnormal{BPP} there is no polynomial-time $o(\log\Delta)$ approximation algorithm for either problem.
\end{theorem}

Complementing this lower bound, we give a matching polynomial-time $O(\log\Delta)$-approximation algorithm for both problems.
\begin{theorem}
\label{thm1.2}
    There is a polynomial-time $O(\log\Delta)$-approximation algorithm for \textnormal{incremental DS} and \textnormal{incremental CDS}.
\end{theorem}
\subsection{Our Techniques}
We start in Section~\ref{sec:LP} by formulating both incremental DS and incremental CDS as covering problems. The formulation of incremental DS is rather standard, but for incremental CDS we introduce a more involved formulation that must maintain connectivity throughout. For each newly requested vertex, we create different constraints depending on whether it connects to zero, one, or more than one existing connected component. This formulation simultaneously captures both domination and connectivity requirements, while ensuring that both the maximum constraint size and maximum frequency are linear in $\Delta$, enabling the application of standard techniques for online covering problems. This formulation is the foundation for most of our results.

As an immediate result from the linear formulation we prove the $O(\log^2\Delta)$-competitive ratio for online randomized algorithms, and an $O(\log\Delta)$ approximation ratio for the offline setting, where both the graph and request order are known in advance. 
Another result of this formulation is a greedy deterministic $O(\Delta)$-competitive algorithm for the weighted variant of both settings.

In Section~\ref{sec:neighbor}, we use the extra information provided by knowing the neighborhood of requested vertices to exponentially improve the competitive ratio for deterministic algorithms. When the graph is unweighted, we maintain both a fractional and an integral solution, coupled through a potential function that is exponential in the fractional coverage of elements not yet covered integrally. The exponential growth ensures that at most $O(|\text{OPT}|)$ elements can be covered fractionally without being covered integrally. We combine this bound with an algorithm that guarantees that whenever a new uncovered vertex is requested, we add a logarithmic number of vertices to the dominating set. Together, these techniques provide the desired competitive ratio. We note that our potential function approach is inspired by Alon et al.~\cite{osc}, yet our solution gives improved competitive bounds and handles a wider range of settings.

We note that $\Delta_G$ can increase over time, to handle this we use a guess-and-square technique that achieves the desired $O(\log^2 \Delta_G)$ competitive ratio, improving upon the $O(\log^3 \Delta_G)$ bound obtained by the standard guess-and-double approach (see Section~\ref{doubling} for more details). This technique is used several times throughout the paper for algorithms where the maximum degree of the graph is not known in advance.

To get a similar competitive ratio for incremental CDS, we partition constraints into two sets: the first one contains the constraints handled by the incremental DS algorithm and the second one contains at most $O(|\text{OPT}|)$ other constraints. The small size of the second set ensures that satisfying these constraints preserves the competitive ratio.

For the weighted case, we augment the potential function with a term that bounds the difference between the cost of the integral and fractional solutions while simultaneously bounding the increase of the potential throughout the runtime, yielding a logarithmic competitive ratio for incremental DS.

To obtain matching lower bounds on our competitive ratios, we provide reductions from dominating set and online set cover. A key insight in both reductions is the use of auxiliary vertices that are revealed first and must be selected. Then, we reveal vertices which are covered by these auxiliary vertices, revealing information about the instance to the incremental algorithm (either the sets or the graph depending on the reduction), thus enabling us to construct the reductions. These reductions preserve both the instance size and the optimal solution size, allowing us to transfer known lower bounds to our setting. 

For the random-order lower bound, we replicate vertices so that a random permutation induces the adversarial set cover sequence with high probability, transferring random-order set cover lower bounds to the random order setting.

 \section{Preliminaries}
We introduce the notation used throughout the paper, formally defining the problems we study, and presenting their covering formulations.

 Let $G = (V, E, w)$ denote an undirected graph with vertex set $V$, edge set $E$, and weight function $w: V \to \R^+$. Let $w(S) = \sum_{s\in S}w(s)$ denote the total weight of all elements in $S\subseteq V$. 
The set of all connected components of $G$ is denoted by $C(G)$. 

Without loss of generality, we assume that $V = \{v_1, \dots, v_n\}$, and that vertices are requested in this order, i.e., vertex $v_1$ is requested first, then $v_2$, and so on.  Denote $V_i = \{v_1,\dots, v_i\}$ when the order is understood from context.

For a vertex $v$, we denote by $N(v)$ the set of neighbors of $v$ and $N[v] = N(v)\cup \{v\}$. We define the natural extensions to sets $N[U] = \bigcup_{u\in U} N[u]$.

For any $H\subseteq V$, let $G[H]$ be the induced subgraph of $G$ on $H$, formally $G[H] = (H, E\cap H\times H, w|_H)$ where $w|_H$ is the restriction of $w$ to $H$. This operation removes all vertices which are not in $H$, as well as edges incident to them.

We now define the competitive ratio used to benchmark our algorithms 
throughout the paper.

\begin{definition}[Competitive Ratio]
The \textnormal{competitive ratio} of an online algorithm $\mathcal{A}$ is 
\[\sup \frac{w(\mathcal{A})}{w(\mathrm{OPT})},\]
where $\mathrm{OPT}$ denotes the cost of the optimal algorithm, and the supremum is taken 
over all graphs, weight functions, and request sequences.
\end{definition}

In this paper, OPT always refers to an optimal \textbf{incremental} algorithm.

\subsection{Problem Definitions}
We begin by defining the dominating set and connected dominating set problems.
\begin{definition}[Dominating Set]
    For a graph $G$, a \textnormal{dominating set} is a subset $S\subseteq V$ such that $\bigcup_{v\in S}N[v] = V$. The objective is to minimize $w(S)$. 
\end{definition}
In addition to covering all vertices, we may impose the requirement that the dominating set be connected within each connected component of the graph. This leads to the following definition.
\begin{definition}[Connected Dominating Set]
A \textnormal{connected dominating set (CDS)} is a \textnormal{dominating set} of $G$ such that for every connected component $C$ of $G$, the set $S \cap C$ is connected.
\end{definition}
To model the decisions an online algorithm must make incrementally, we define the incremental dominating set problem and its connected variant.
Recall that $V_i$ denotes the set containing the first $i$ vertices requested by the online process.
\begin{definition}[Incremental Dominating Set]
For a graph $G$, an \textnormal{incremental dominating set} is a subset $S\subseteq V$ such that for all $i\in[n]$, $S\cap V_i$ is a \textnormal{dominating set} of $G[V_i]$.
\end{definition}
\begin{definition}[Incremental Connected Dominating Set]
For a graph $G$, an \textnormal{incremental connected dominating set} is a subset $S\subseteq V$ such that for all $i\in[n]$, $S\cap V_i$ is a \textnormal{connected dominating set} of $G[V_i]$.
\end{definition}

In the online setting, the vertices of a graph are requested one by one. In this paper, we consider two information models: at step $i$, the algorithm either knows the induced subgraph $G[V_i]$, which corresponds to the standard vertex arrival model, or the neighborhood subgraph $G[N[V_i]]$. Throughout the paper, the competitive ratio is measured against $\mathrm{OPT}$, the optimal incremental dominating set.

\begin{definition}[Known Neighborhood Model]
In the \textnormal{known neighborhood model}, at step $i$ when vertex $v_i$ is requested, the algorithm observes the induced subgraph $G[N[V_i]]$ rather than just $G[V_i]$. Vertices in $N[V_i] \setminus V_i$ are visible but cannot be selected for the dominating set until requested.
\end{definition}

We work in the \textnormal{late-accept} model~\cite{main, boyar2022relaxing}, where vertices can be added to the solution at any time but cannot be removed once chosen. We now define the online variants of the incremental Dominating Set and incremental Connected Dominating Set problems.

\begin{definition}[Online Incremental Dominating Set]
In the \textnormal{online incremental dominating set} problem, the algorithm must maintain a dominating set of $G[V_i]$ at each step $i$.
\end{definition}

\begin{definition}[Online Incremental Connected Dominating Set]
In the \textnormal{online incremental connected dominating set} problem, the algorithm must maintain a connected dominating set of $G[V_i]$ at each step $i$.
\end{definition}
Since any superset of a CDS is also a CDS, the late accept setting does not violate the incrementality requirements of previous steps. Adding vertices to the solution at later steps preserves the CDS property for all previous graphs.
   
   \subsection{Covering Problems}
   \label{sec:cover}
In this paper, we use linear programs with covering constraints
\begin{definition}[Covering Problem]
    
A \textnormal{covering problem} is a linear program of the form:
\begin{equation}
	\begin{aligned}
  ~~~~~~~~	& \min~~~  \sum_{i\in [n]}  w(i)\cdot x_i \\
		& ~~~~~~~~\text{s.t. }\\
		& a^j\cdot \vec{x}  \geq 1, 	~~~~~~~~~~~~~~\forall j\in [m] \\ 
		& ~~~~x_i \geq 0 ~~~~~~~~~~~~~~~\forall i \in [n]
	\end{aligned}
\end{equation} 
where $a^j \in \{0,1\}^n$ for all $j\in [m]$.
\end{definition}

\section{Covering LPs for Incremental DS and Incremental CDS}\label{sec:LP}
We now formulate linear programming (LP) relaxations for incremental DS and incremental CDS, expressing both as covering problems with each constraint involving at most $O(\Delta)$ variables and each variable participating in $O(\Delta)$ constraints. Formulating both problems as cover LPs has two main advantages. First, it allows us to leverage techniques already known in the literature for online covering problems. Second, it provides a unified framework for designing algorithms for both problems simultaneously, an approach we exploit in later sections.

We formalize these observations in the following theorem,
\begin{theorem}\label{thm3.1}

Both \textnormal{incremental DS} and \textnormal{incremental CDS} have \textnormal{LP} relaxations that are covering problems. The size of each constraint is bounded, and so is the number of variables in each constraint. More precisely:
    \begin{enumerate}
        \item For \textnormal{incremental DS}, each constraint involves at most $\Delta + 1$ variables.\label{thm3.1a}
    \item For \textnormal{incremental CDS}, each constraint involves at most $\Delta$ variables.\label{thm3.1b}
    \item In both relaxations, each variable appears in at most $\Delta + 1$ covering constraints, and a single non-negativity constraint.
    \end{enumerate}
\end{theorem}
\subsection{Incremental DS as a Covering Problem}
We begin by formulating the LP relaxation for incremental DS. For each vertex $v_i$, we introduce a variable $x_i \in [0,1]$.
\begin{proof}[Proof of Theorem~\ref{thm3.1}(\ref{thm3.1a})]
When $v_i$ is requested, either $v_i$ or one of its (previously requested) neighbors needs to be in the dominating set, so 
\[ \sum_{j: v_j \in N[v_i] \cap V_i} x_j \ge 1.\]
Due to the monotone nature of the solution, it suffices to ensure that each newly requested vertex is covered when it appears. This gives us the following relaxation
\begin{equation}
	\label{eq:LP}
	\begin{aligned}
  ~~~~~~~~	& \min~~~  \sum_{i\in [n]}  x_i \cdot w(v_i)\\
		& ~~~~~~~~\text{s.t. }\\
		& \sum_{j: v_j\in N[v_i] \cap V_i }x_j  \geq 1, 	~~~~~~~~~\forall i\in [n] \\ 
		& ~~~~~~~~~~~~~~~~x_i \geq 0, ~~~~~~~~~\forall i \in [n]
	\end{aligned}
\end{equation} 
This relaxation is a covering problem. Each constraint involves at most $\Delta + 1$ variables since it includes $v_i$ and its neighbors in $V_i$.
Moreover, each variable $x_j$ appears only in the non-negativity constraint and constraints for vertices $v_i$ such that $v_j\in N[v_i]$. Since $|N[v_i]| \le \Delta + 1$, each variable $x_j$ appears in at most $\Delta + 1$ covering constraints. 
\end{proof}
\subsection{Incremental CDS as a Covering Problem}\label{sec:CDSascover}

In this subsection, we prove Theorem~\ref{thm3.1}(\ref{thm3.1b}) by identifying a small set of integral constraints that ensures the chosen subset is both a dominating set and is connected. We provide the intuition for the LP formulation here, while the formal proof of correctness is deferred to Appendix~\ref{sec:cdstolp}.

To express connectivity using local constraints, we classify each new vertex by how it interacts with existing components: it may start a new one, attach to one, or join several together. The cases require different types of constraints to maintain both domination and connectivity. To the best of our knowledge, this is the first time that  connectivity is captured using only local constraints.

\begin{proof}[Proof of Theorem~\ref{thm3.1}(\ref{thm3.1b})]
We start by classifying the vertices into three types:

\begin{enumerate}
    \item Opener - A vertex that, upon arrival, is not connected to any other vertex. Formally, \(N[v_i] \cap V_i = \{v_i\}\). For example, the first vertex \(v_1\) is always an opener.
    \item Attacher - A vertex that, upon arrival, is connected to only one component. Formally, there exists \(C \in C(G[V_{i-1}])\) such that for all \(C' \neq C\), \(N[v_i]\cap V_i \cap C' = \emptyset\).
    \item Joiner - A vertex that connects at least two components upon arrival. Formally, there exist \(C_1, C_2 \in C(G[V_{i-1}])\) such that \(N[v_i] \cap V_i \cap C_1 \neq \emptyset\) and \(N[v_i] \cap V_i \cap C_2 \neq \emptyset\).
\end{enumerate}
We now add the appropriate constraints for each type to enforce both domination and connectivity.
\paragraph{Case 1: Opener.}
An Opener creates a new component and must be chosen to dominate itself. The only way to satisfy the dominating set requirement is to include it in the dominating set, thus $x_i \geq 1$.

\paragraph{Case 2: Attacher}
An Attacher connects to one existing component $C$. The domination and connectivity constraints are satisfied as follows:
\begin{itemize}
    \item If \(v_i\) is added to the CDS, at least one of its neighbors in \(C\) must also be in the CDS to maintain connectivity.
    \item If \(v_i\) is \textbf{not} added to the CDS, at least one of its neighbors must be added to satisfy the domination requirement.
\end{itemize}
Both scenarios are enforced by the constraint:
\[
\sum_{v_j \in N(v_i) \cap V_i} x_j \geq 1.
\]
This constraint is sufficient because the neighbors of \(v_i\) are already adjacent to the CDS, so adding any neighbor maintains both domination and connectivity constraints.

\paragraph{Case 3: Joiner.}
The joiner connects multiple components \(C_1, \dots, C_k\). After adding the joiner the dominating set of any pair of components must have a path between them, but such a path must go through the joiner. As a result the joiner must be added to the dominating set. To maintain connectivity:
\begin{itemize}
    \item \(v_i\) must be included in the CDS, $x_i \geq 1$.
    \item For each component \(C_h\), at least one neighbor of \(v_i\) within \(C_h\) must be included:
    \[
    \sum_{v_j \in N(v_i) \cap V_i \cap C_h} x_j \geq 1.
    \]
\end{itemize}
This requirement is sufficient since any other vertex is already adjacent to the CDS, so adding a neighbor keeps the joiner connected to the dominating set in $C_h$.

The maximum size of a constraint is $\Delta$ for attachers.
Moreover, each constraint corresponds to a vertex \(v_i\) and involves only its neighbors. Any variable \(x_j\) appears in one non-negativity constraint, at most one constraint for each neighbor \(v_i\), and at most one constraint for $v_j$ itself. Since each vertex has at most \(\Delta\) neighbors, each variable appears in at most $\Delta + 2$ constraints overall.
\end{proof}
\subsection{Algorithmic Implications}
We start by showing the strength of Theorem~\ref{thm3.1} by combining it with several known theorems about covering problems to obtain immediate results.

The first of such results is for the offline problem, where from Srinivasan~\cite{srinivasan1999improved} we have the following theorem.
\begin{theorem}[Srinivasan~\cite{srinivasan1999improved}]
\label{thm:srinivasan-offline}
For any covering problem with sparsity $d$ (i.e., each variable appears in at most $d$ constraints), there is an $O(\log d)$-approximation algorithm.
\end{theorem}

Combining Theorems~\ref{thm3.1} and~\ref{thm:srinivasan-offline} yields the approximation ratio for the offline setting of both incremental DS and incremental CDS.

\newtheorem*{manualThm110}{Theorem 1.10}
\begin{manualThm110}
    There is a polynomial-time $O(\log\Delta)$-approximation algorithm for \textnormal{incremental DS} and \textnormal{incremental CDS}.
\end{manualThm110}
\begin{proof}
By Theorem~\ref{thm3.1}, both incremental DS and incremental CDS can be formulated as covering integer programs where each variable appears in at most $\Delta + 1$ constraints, meaning that the column sparsity is $d = \Delta + 1$. 
Applying Theorem~\ref{thm:srinivasan-offline} yields an $O(\log(\Delta + 1)) = O(\log \Delta)$-approximation algorithm for both problems.
\end{proof}

For the online setting, we apply the following theorem from Gupta and Nagarajan~\cite{gupta2014approximating}.

\begin{theorem}[Gupta and Nagarajan~\cite{gupta2014approximating}]
\label{thm:gupta-online}
There exists an $O(\log k \cdot \log d)$-competitive randomized online algorithm for covering problems, where $k$ is the maximum number of variables in any constraint and $d$ is the sparsity (maximum number of constraints any variable appears in).
\end{theorem}

Again, combining Theorems~\ref{thm3.1} and~\ref{thm:gupta-online} immediately yields the randomized competitive ratio.

\newtheorem*{manualThm12}{Theorem 1.2}
\begin{manualThm12}
There is an $O(\log^2\Delta)$-competitive randomized algorithm for \textnormal{incremental DS} and \textnormal{incremental CDS}.
\end{manualThm12}
\begin{proof}
By Theorem~\ref{thm3.1}, the LP relaxations of both problems are online covering problems with row sparsity $k \leq \Delta + 1$ and column sparsity $d \leq \Delta + 1$. Applying Theorem~\ref{thm:gupta-online} gives an $O(\log^2 \Delta)$-competitive randomized online algorithm.
\end{proof}

\subsubsection{Greedy Deterministic Online Algorithm}
\label{sec:algorithm}
We now a present deterministic algorithm addressing the weighted variant of both online incremental DS and online incremental CDS. Throughout the section, we rely on Theorem~\ref{thm3.1} to design algorithms that solve both of these problems. We describe a deterministic greedy online algorithm achieving a competitive ratio of $O(\Delta)$ for both problems. This bound is tight up to constant factors, matching the lower bound proven in Boyar et al.~\cite{main}.

Whenever a new vertex is requested it generates new constraints. When constraint \(j\) arrives the algorithm sets the variable \(x_{g(j)}\) to $1$, where $g(j)$ is the variable with the smallest cost \(w(g(j))\) among all variables that satisfy the constraint. 
We break ties arbitrarily. If the selected variable \(g(j)\) has already been added to the dominating set, the algorithm takes no action.

Let $M$ be the column sparsity, i.e., the maximum number of covering constraints in which any variable appears. For each constraint \(j\), let \(g(j)\) be the variable chosen by the greedy algorithm, and let \(opt(j)\) be a variable that satisfies the same constraint in the optimal solution. By the greedy choice, \(w(g(j)) \le w(opt(j))\).

\begin{lemma}\label{lem4.1}
The greedy algorithm is $M$-competitive.
\end{lemma}
\begin{proof}
The total cost of the greedy algorithm is
\[
\sum_{j} w(g(j)) \le \sum_{j} w(opt(j)).
\]
Since each variable can appear in at most \(M\) covering constraints, each cost term of the optimal solution is counted at most \(M\) times on the right-hand side. Hence,
\[
\sum_{j} w(opt(j)) \le M \cdot \text{OPT},
\]
where OPT is the cost of the optimal solution, and the greedy algorithm is \(M\)-competitive.
\end{proof}

We now prove our main theorem about the competitive ratio of deterministic algorithms for incremental DS and incremental CDS.

\newtheorem*{manualThm11}{Theorem 1.1}
\begin{manualThm11}
There is a $(\Delta + 1)$-competitive deterministic algorithm for \textnormal{incremental DS} and \textnormal{incremental CDS}.
\textnormal{incremental DS}.
\end{manualThm11}
\begin{proof}
By Lemma~\ref{lem4.1}, the greedy algorithm is \(M\)-competitive. By Theorem~\ref{thm3.1}, in both incremental DS and incremental CDS we have \(M \le \Delta + 1\). This yields a $\Delta + 1$ competitive algorithm for deterministic incremental DS and incremental CDS.
\end{proof}

\section{Know Thy Neighbor}\label{sec:neighbor}
In this section we present improved deterministic algorithms for the case where the neighborhood of each requested vertex $v_i$ ($N[V_i]$) is revealed when $v_i$ needs to be dominated. The setting remains incremental and vertices which have not been requested (yet were revealed earlier) cannot be used to cover requested vertices. In other words, when a vertex is requested, we learn about its neighbors, but only the requested vertex itself needs to be dominated.

Knowing the neighborhood of each requested vertex complies with the arrival model studied in previous sections, except that additional information is provided at runtime. This extra information allows us to design deterministic algorithms with polylogarithmic competitive ratios, a substantial improvement over the $\Omega(\Delta)
$ lower bound from the setting where only requested vertices are known (without full neighborhood information).

Since revealed vertices need not all be requested, we distinguish between $\Delta_G$, the maximum degree of the full revealed graph, and $\Delta$, the maximum degree of $G[V_n]$, the subgraph induced by requested vertices. 

In Section~\ref{sec:unweighted}, we present an $O(\log^2 \Delta_G)$-competitive deterministic algorithm for the unweighted case, for both incremental DS and incremental CDS. Then, in Section~\ref{sec:gen}, we show that in the weighted case there exists an $O(\log n \log \Delta)$-competitive deterministic algorithm for incremental DS. We leave improving the upper bound for weighted incremental CDS as an open problem.

This model differs from that of~\cite{harutyunyan2021online}, where a vertex must be covered only when its last neighbor is revealed, and late acceptance is disallowed.

\subsection{Unweighted Case}\label{sec:unweighted}
We first present a deterministic algorithm for incremental DS. Our approach is based on maintaining both a fractional solution and an integral dominating set $D$, coupled through a potential function. The potential function penalizes vertices that are covered fractionally but not integrally, ensuring that vertices with high fractional coverage are added to the integral solution.

The algorithm maintains a parameter $m$, initialized to an upper bound on the maximum degree of the revealed graph
$\Delta_G$, and updated dynamically to remain an upper bound as new vertices 
arrive (see Section~\ref{doubling}). We also maintain a fractional solution $x$ (possibly infeasible), initialized with $x_i = \tfrac{1}{4m}$ for each vertex $v_i$. 

To track how well each vertex is covered fractionally, we define the potential cover $c_v$ for each known vertex $v$. This measure has two components: the fractional coverage from revealed neighbors, and a correction term that accounts for potential neighbors that have not yet been revealed. Specifically, when a new neighbor of $v$ is revealed, it contributes $\tfrac{1}{4m}$ to the fractional coverage, while simultaneously decreasing the correction term by the same amount, ensuring $c_v$ remains unchanged. Formally, for each revealed vertex $v$ we define:
\[c_v = \sum_{u_j\in N[v]\cap V_i} x_j + \frac{m-|N[v]\cap V_i|}{4m}.\]

Let $D$ be the dominating set maintained by the algorithm at the current 
step, and denote the set of uncovered elements as $U_i = N[V_i]\setminus N[D]$. Define the potential at step $i$ as $\Phi_i = \sum_{v\in U_i}m^{3c_v}$. The exponential term ensures that vertices with high fractional coverage cannot remain uncovered integrally.

At each step, we first initialize any newly revealed neighbors of $v_i$ with 
fractional value $\tfrac{1}{4m}$. If $v_i$ is uncovered, we snapshot 
the current potential as $\Phi_s$, then increase the weights of vertices 
in $N[v_i]\cap V_i$ until $v_i$ is covered fractionally. We then add the 
smallest subset of $N[v_i]\cap V_i$ to $D$ such that the potential is 
restored to at most $\Phi_s$; by Lemma~\ref{lem:subset} this subset has size 
$O(\log m)$. Finally, if $v_i$ remains uncovered after this process, 
we add $v_i$ itself to $D$. A pseudocode of the algorithm is described in 
Algorithm~\ref{alg:det}.

At each step where an uncovered $v_i$ is requested, we increase the weights of vertices in $N[v_i]\cap V_i$ until $v_i$ is covered in the fractional solution. We then add $O(\log m)$ vertices from $N[v_i]\cap V_i$ to the dominating set so that $\Phi_i \le \Phi_{i-1}$.
Finally, if $v_i$ is still not covered, we add it to the dominating set $D$ to ensure its viability. The algorithm is described in Algorithm~\ref{alg:det}. 

\begin{algorithm}[ht]
	\caption{Unweighted Known Neighborhood Deterministic}
	\label{alg:det}
    \tcp*[h]{Notation as defined in the preceding text.}\\
Initialize $D \gets \emptyset$ (initial dominating set).

	\For{vertex $v_i \in V$ }{\label{algl3}
        For each newly revealed vertex $v_j$, set $x_j\gets \frac1{4m}$, and increment $\Phi_i$ by $m^{3c_{v_j}} \le m^{3/4}$.\label{alg2l4}
        
        \If{$v_i\in U_i$}{\label{alg2l5}
            Set $\Phi_s\gets \Phi_i$.
            
            \While{$c_{v_i} < 1$}{\label{alg2l6}
                For each $v_j\in N[v_i]\cap V_i$, multiply $x_j$ by two.\label{alg2l7}
            }
            Add a the smallest subset of vertices to $D$ for which $\Phi_i \le \Phi_s$. // By Lemma~\ref{lem:subset}, the size of such a subset is $O(\log m)$.\label{alg2l9}
        }
        \If{$v_i$ is not covered\label{alg2l10}} {Add $v_i$ to $D$.\label{alg2l11}}
        }
\end{algorithm}

To begin the analysis, we show that after each weight increase the algorithm can avoid an increase in the potential function by adding only a small subset of vertices to the dominating set in Line~\ref{alg2l9}.

\begin{lemma}\label{lem:subset}
If $v_i$ is uncovered at step $i$, then there exists a subset of 
$N[v_i]\cap V_i$ of size at most $6\log m$ such that adding the vertices 
in the subset to $D$ ensures that the potential does not increase, i.e.,
\[
\Phi_i \le \Phi_s,
\]
where $\Phi_s$ is the potential prior to the weight increases in Line~\ref{alg2l7}.
\end{lemma}
The proof is deferred to Appendix~\ref{sec:omit}. 
We next move from considering how much is added per update to bounding how often such updates happen. 

\begin{lemma}\label{lem5.1}
The total number of times the {\bf if} statement in Line~\ref{alg2l5} in Algorithm~\ref{alg:det} is triggered is at most $O(|\textnormal{OPT}|\log m)$.
\end{lemma}
\begin{proof}
Every coordinate $x_j$ is initialized to $1/(4 m)$ and the algorithm only multiplies coordinates by $2$ (Line~\ref{alg2l7}). Moreover, Line~\ref{alg2l6} ensures that the algorithm never increases the value of any coordinate beyond $2$. Thus, for every index $j$ the number of times $x_j$ can be doubled is at most
$\log(8m) = O(\log m)$.

Since $v_i$ must be dominated by a vertex in OPT, each time the \textbf{if} statement is triggered at least one coordinate corresponding to a vertex of OPT is doubled.

Combining these results, the total number of times that the \textbf{if} statement can be triggered is upper bounded by $O(|\textnormal{OPT}|\cdot \log m)$.
\end{proof}

Having bounded the number of update steps, we now show that the increase in the potential function itself remains bounded throughout the execution. Since each vertex added in the final while loop decreases the potential by a large amount, bounding the total potential increase immediately implies a bound on the number of vertices added in that phase.

\begin{lemma}\label{leminc}
At step $i$, $\Phi_i \le i\cdot \Delta_G\cdot m$.
\end{lemma}
\begin{proof}
    The potential can change in three places: when a vertex is requested, when weights are increased in Line~\ref{alg2l7}, and when new vertices are revealed. 

    First, consider the moment when a vertex $v_i \in N[u]$ is requested. 
Its initial contribution of $1/(4m)$ is added to the first summand of 
$c_u$, but at the same time $|N[u]\cap V_i|$ increases by one, which 
decreases the correction term $\left(m-|N[u]\cap V_i|\right)/4m$ by exactly 
$1/(4m)$. The two changes cancel, and $c_u$ does not change.

    Second, by Lemma~\ref{lem:subset}, after each weight increase in 
    Line~\ref{alg2l7} the vertices added to $D$ in Line~\ref{alg2l9} 
    ensure that the potential does not increase, so weight increases cause 
    no net increase in the potential.

    Therefore the only source of increase is when new vertices are revealed. For each new vertex $u$, the increase in the potential is at most $m^{3c_u} \le m$, since $c_u \le \frac14$ at reveal. After $i$ steps, the amount of revealed vertices is bound by $i\cdot \Delta_G$, the total increase is bounded by $i\cdot \Delta_G\cdot m$.
\end{proof}

We now use the bounded increase in $\Phi$ to bound the number of vertices added in the final loop of the algorithm. Intuitively, since each additional vertex in that phase reduces the potential by a large amount, a global bound on the total potential increase immediately limits how many such insertions can occur.

\begin{lemma}\label{lem5.2}
    Denote the number of total requested vertices throughout the algorithm as $r$. Then number of vertices added by Lines~\ref{alg2l10}-\ref{alg2l11} is bounded by $\frac{r\cdot \Delta_G}{m^2}$, and if $m\ge \Delta_G$ then the bound is $2\cdot |\textnormal{OPT}|$.
\end{lemma}
\begin{proof}
When we insert $v_i$ to $D$ in Line~\ref{alg2l11}, we remove at least $m^{3c_u}\ge m^3$ from the potential. By lemma~\ref{leminc} the number of added vertices is at most
\[\frac{r\cdot\Delta_G\cdot m}{m^3} = \frac{r\cdot \Delta_G}{m^2}\]
If $m \ge \Delta_G$ then $|\textnormal{OPT}|\ge r/(\Delta_G+1)\ge r/(2\Delta_G)$, giving us the upper bound.
\end{proof}

We are now ready to combine the previous bounds to obtain the overall competitive ratio.

\begin{lemma}\label{lem5.5a}
    If $m \ge \Delta_G$, the competitive ratio of algorithm~\ref{alg:det} is $O(\log^2 m)$.
\end{lemma}
\begin{proof}
By Lemma~\ref{lem5.1} we enter the \textbf{if} in Line~\ref{alg2l5} at most $O(|\textnormal{OPT}|\cdot\log m)$ times, and by Lemma~\ref{lem:subset} each time we add $4\log m$ vertices to $D$. By Lemma~\ref{lem5.2} we then add at most $2|\textnormal{OPT}|$ vertices in Lines~\ref{alg2l10}-\ref{alg2l11}. Hence
    \[
    |D| \;\le\; O(|\textnormal{OPT}|\cdot\log m)\cdot 4\log m + 2|\textnormal{OPT}| \;=\; O(|\textnormal{OPT}|\cdot\log^2 m),
    \]
    proving the claimed competitive ratio.
\end{proof}
We note that Algorithm~\ref{alg:det} extends naturally to the case where 
constraints are of the form $N(v_i) \cap D$ (as with 
attacher constraints), where the requested vertex itself is excluded. 
If $N(v_i) = \emptyset$ then $v_i$ is an isolated vertex and no constraint 
is enforced. Otherwise, all we need to do is replace all instances 
of $N[v_i]$ with $N(v_i)$. The potential function argument is unchanged, 
and the same $O(\log^2\Delta_G)$-competitive bound holds.

\subsubsection{Connected Dominating Set}
In this part we extend the result of Section~\ref{sec:unweighted} to 
incremental CDS. The algorithm is similar to Algorithm~\ref{alg:det}, but with a special case for joiners (recall that a joiner is a vertex that joins two connected components, see Section~\ref{sec:CDSascover} for the formal definition). When a joiner is requested, for each constraint created by it, we immediately add an arbitrary vertex from the constraint to the dominating set. We prove that the number of vertices added this way is bounded by $O(|\textnormal{OPT}|)$.

\begin{lemma}\label{lem5.5}
    The number of constraints added by joiners is bounded by $2\cdot |\textnormal{OPT}|$.
\end{lemma}
\begin{proof}
    We actually prove a stronger statement, namely that the total number of constraints created by openers and joiners is at most $2\cdot |\textnormal{OPT}|$.
    Let $\mathcal{O}$ be the number of openers and let $\mathcal{J}$ be the number of joiners.
    By the transformation in Section~\ref{sec:CDSascover}, every opener and joiner must belong to the dominating set, so $|\textnormal{OPT}| \ge \mathcal{O} + \mathcal{J}$.

    We use a charging argument. When an opener is requested, it becomes the representative of its component. When a joiner is requested, each connected component it merges creates one constraint, which we charge to that component's representative. The joiner then becomes the representative of the merged component.

Each vertex can be charged at most twice: once when requested (as an opener or joiner), and at most once later as a representative. Therefore, the total number of constraints is at most $2(\mathcal{O} + \mathcal{J}) \le 2\cdot |\textnormal{OPT}|$.
\end{proof}

\subsubsection{Guess and Increase}\label{doubling}

To extend our algorithm to the case where the maximum degree $\Delta_G$ is not known in advance, we use a guess-and-update technique inspired by the standard “guess-and-double” framework. A straightforward adaptation, starting with an underestimate of $\Delta_G$ and doubling it whenever the maximum degree exceeds the current guess, leads to a competitive ratio of
\[\sum_{j = 1}^{\log\Delta_G}O(\log^2(\Delta_G/2^j))\cdot \textnormal{OPT} = O(\log^3\Delta_G)\cdot \textnormal{OPT},\]which is asymptotically worse than the optimal bound. To handle this discrepancy, instead of doubling the current estimate of $\Delta_G$, we square it upon discovering that the maximum degree violates the current guess.

We are now ready to prove the main theorem of the section.
\newtheorem*{manualThm13}{Theorem 1.3}
\begin{manualThm13}
When the neighborhood of requested vertices is known, and the graph is unweighted, then
there is an $O(\log^2\Delta_G)$-competitive deterministic algorithm for both \textnormal{incremental DS} and \textnormal{incremental CDS}.
\end{manualThm13}
\begin{proof}
Denote by $\text{OPT}_i$ the cost of the optimal solution for the first $i$ vertices, and let $\Delta_i$ be the maximum degree of the revealed graph at this step. Since each new vertex introduces additional constraints, it follows that $\text{OPT}_i \le \text{OPT}$ for all $i$.

We maintain an adaptive upper bound $\Delta'$ on the current maximum degree, initialized to a constant. As vertices arrive, whenever $\Delta_i > \Delta'$, we update $\Delta' \leftarrow (\Delta')^2$, restart the algorithm with the new bound, and include in the global solution all vertices selected during the new run while retaining vertices selected during previous solutions.

By Lemma~\ref{lem5.5a}, each phase with bound $\Delta'$ has cost
$O(\log^2 \Delta')\cdot \textnormal{OPT}$. Since $\Delta'$ grows quadratically, after $k$ such increases the total cost is bounded by
\[O\left(
   \sum_{j=0}^k \log^2\!\big(\Delta_G^{1 / 2^{\,k-j}}\big)
\right)\cdot\textnormal{OPT} \le 
O\left(
   \sum_{j=0}^\infty \log^2\!\big(\Delta_G^{ 2^{\,-j}}\big)
\right)\cdot\textnormal{OPT}
= O(\log^2\!\Delta_G)\cdot\textnormal{OPT}.
\]
Therefore, the adaptive algorithm preserves the $O(\log^2 \Delta_G)$
competitive ratio even when $\Delta_G$ is unknown in advance.

To extend the result to incremental CDS, we use the CDS formulation from Section~\ref{sec:CDSascover}.
We use Algorithm~\ref{alg:det} 
to handle the attacher constraints (as defined in Section~\ref{sec:CDSascover}, these are the constraints corresponding to vertices that attach to an existing component) and achieve an $O(\log^2\Delta_G)$-competitive 
dominating set for these constraints. Since the optimal solution for the 
attacher constraints is at most $|\textnormal{OPT}|$, the number of 
vertices selected by Algorithm~\ref{alg:det} is bounded by $O(\log^2\Delta_G)\cdot|\textnormal{OPT}|$. By Lemma~\ref{lem5.5}, there are at most $2|\textnormal{OPT}|$ constraints 
created by openers and joiners, and adding one vertex per such constraint 
contributes at most $2|\textnormal{OPT}|$ additional vertices. Combining 
both contributions results in a valid incremental CDS of size 
$O(\log^2\Delta_G)\cdot|\textnormal{OPT}|$, finishing the proof.
\end{proof}

\subsection{General Algorithm for Incremental DS}\label{sec:gen}
In this part, we show an $O(\log m\log \Delta)$ deterministic algorithm for incremental DS that uses an upper bound on the value of the optimal solution $\alpha$, and a parameter $m$ which represents an upper bound on the number of \textbf{requested} vertices $n$. In this section $\Delta$ represents the degree of $G[V_n]$, the subgraph induced by requested vertices, disregarding vertices which are revealed but not requested. 
We apply doubling to handle the case where $\alpha$ is not known, and we later show how to increase 
$m$ dynamically, using a technique similar to doubling, to obtain the desired competitive ratio of $O(\log n \log \Delta)$. If the number of requested vertices $n$ is known in advance we can set $m = n$ instead.

Our algorithm maintains both a fractional and an integral solution. We use the fractional solution to guide which vertices to add to the integral dominating set $D$, while a potential function ensures the two solutions remain balanced. To obtain the fractional solution, we apply the following result from Buchbinder et al.~\cite{book}:
\begin{theorem}\label{thm4.3}
    There is a deterministic $O(\log d)$-competitive fractional algorithm for covering problems with constraint size up to $d$.
\end{theorem}
Since each covering constraint involves at most $\Delta + 1$ variables, Theorem~\ref{thm4.3} guarantees an online monotone fractional solution which is $O(\log\Delta)$ competitive.

Denote the dominating set as $D$. Let the set of vertices adjacent to requested vertices that are not yet dominated at step 
$i$ be $U_i = N[V_i]\setminus N[D]$. Denote $c_{v_i} = \sum_{u_j\in N[v_i]\cap V_i} x_j$ as the fractional coverage of $v_i$, and define the potential at step $i$ as \[\Phi_i = \sum_{v\in U_i}m^{2c_v} + \exp\left(\frac1{2\alpha}\left(\sum_{v\in D}w(v) - 3\log m\sum_{v_j\in V_i}w(v_j)x_j\right)\right).\]
Where the first sum ensures that every vertex is covered when it is requested, and the second term ensures that the integral and fractional solutions do not differ by too much.

At each step a vertex $v_i$ is requested to be dominated. We update the potential function for the newly revealed neighbors of $v_i$, and then add the covering constraint $\sum_{j \in N[v_i] \cap V_i} x_j \geq 1$ to the LP. We increase the weights of vertices in $N[v_i]\cap V_i$ in the fractional solution until this constraint is satisfied. Then, we add a subset of vertices to $D$ that ensures that the potential does not increase. Vertices $v_j$ with $w(v_j) > \alpha$ are kept at $x_j = 0$ in the fractional solution, ensuring that heavy vertices are never selected. The algorithm is described formally in Algorithm~\ref{alg:weighted}.

\begin{algorithm}[ht]
	\caption{Known Neighborhood Deterministic}
	\label{alg:weighted}
    \tcp*[h]{Notation as defined in the preceding text.}\\

Initialize $D \gets \emptyset$ (initial dominating set).

Initialize $x \gets \vec{0}$, $\Phi_0 \gets e^0 = 1$.

	\For{vertex $v_i \in V$ }{
        For each newly revealed vertex $u$, increment $\Phi_i$ by $m^{2c_u} = 1$.
        
        Add covering constraint for $v_i$ and compute the online fractional solution $\tilde{x}$ from Theorem~\ref{thm4.3}.\label{alg3l4}
        
        \For{$j = 1,\dots, i$}{\label{alg3l5}
            Set $x_j \gets \tilde{x}_j$. // Update fractional solution

If \(\Phi_i\) increased by the change in last line, add \(v_j\) to $D$\label{alg3l6}
        }
        }
\end{algorithm}

We first show that the potential does not increase during the loop starting at Line~\ref{alg3l5}.
\begin{lemma}\label{lem5.4}
Denote $\Phi_i^j$ as the value of the potential function after the $j$th iteration of the loop in Line~\ref{alg3l5}, and $\Phi_i^0$ as the potential at the beginning of the loop. Then,
\[\forall j: \Phi_i^j \le \Phi_i^{j-1}\]
\end{lemma}
The proof is deferred to Appendix~\ref{sec:omit}.

We next upper bound the potential function.
\begin{lemma}\label{lem5.8} At each step $i$,
\[\Phi_i \le n + 1.\]
\end{lemma}
\begin{proof}
    The potential only changes in Lines~\ref{alg3l4} and~\ref{alg3l6}. By Lemma~\ref{lem5.4} the potential can only increase at Line~\ref{alg3l4}. It increases by at most 1 for each newly revealed vertex, giving a total increase of at most $n$. Combining this with the initial value completes the proof.
\end{proof}

We now prove the competitive ratio of the algorithm.
\begin{theorem}\label{thm5.6}
    If $m \ge n$, then at every step $i$, $D$ is an incremental DS with competitive ratio $O(\log m \log \Delta)$.
\end{theorem}
\begin{proof}
    We first show that $D$ is a valid incremental DS. If it is not, then there is an uncovered vertex $v$ with $c_v\ge 1$, so $\Phi \ge m^{2c_v}\ge m^2 \ge n^2$, which contradicts Lemma~\ref{lem5.8}.

    It remains to show that $\sum_{v\in D} w(v) \le O(\alpha \log m \log \Delta)$.
    By lemma~\ref{lem5.8}, $\Phi \le n+ 1 = 2n$ and then
    \begin{align*}
\exp\Bigg\{\frac{1}{2\alpha}\Big(\sum_{v\in D} w(v) - 3 \log m \sum_{v_j \in V_i} w(v_j) x_j \Big)\Bigg\} &\le 2n \\
\frac{1}{2\alpha}\Big(\sum_{v\in D} w(v) - 3 \log m \sum_{v_j \in V_i} w(v_j) x_j \Big) &\le \log 2n \\
\sum_{v\in D} w(v) &\le 2 \alpha \log 2n + 3 \log m \sum_{v_j \in V_i} w(v_j) x_j \\
\sum_{v\in D} w(v) &\le O(\alpha \log m \log \Delta)
\end{align*}
    In the second inequality we take logarithms of both sides. In the last inequality, we used $m \ge n$ and the fact that the fractional solution $\sum_{v_j \in V_i} w(v_j) x_j$ is $O(\log \Delta)$-competitive. 
\end{proof}

We are now ready to prove the main theorem of the section.

\newtheorem*{manualThm14}{Theorem 1.4}
\begin{manualThm14}
When the neighborhood of requested vertices is known, there is an $O(\log n \log \Delta)$-competitive deterministic algorithm for \textnormal{incremental DS}.
\end{manualThm14}
\begin{proof}
By Theorem~\ref{thm5.6}, if $m \ge n$ then Algorithm~\ref{alg:weighted} 
achieves competitive ratio $O(\log m \log \Delta)$. To handle the case 
where $n$ is not known in advance, we use the guess-and-increase technique 
from Section~\ref{doubling}: we start with an initial constant value of 
$m$, and when $i > m$ we square $m$ and rerun Algorithm~\ref{alg:weighted} 
with the updated value. By the same calculation as in the proof of 
Theorem~\ref{thm1.11}, the resulting cost across all phases is bounded by
\[
\sum_{j = 0}^{\infty} O\big(\alpha \log(m^{2^{-j}}) \log \Delta \big) 
= O(\alpha \log m \log \Delta).
\]
At any point $m \le n^2$, so $\log m = O(\log n)$, yielding the desired 
$O(\log n \log \Delta)$ competitive ratio.
\end{proof}
\section{Lower Bounds}
\subsection{Lower Bounds via Online Set Cover}
We present a reduction from the online set cover problem to the online 
incremental Dominating Set problem. The reduction preserves the instance 
size up to polynomial factors and maintains the optimal solution size up 
to a constant additive term. This allows us to transfer known hardness 
results for online set cover and prove Theorems~\ref{thm1.7} 
and~\ref{thmdetLB}.

Given an instance of online set cover with universe 
$U = \{e_1, e_2, \dots, e_n\}$ and collection 
$\mathcal{S} = \{S_1, S_2, \dots, S_m\}$ with $S_i \subseteq U$, 
construct the incremental dominating set instance as follows. 
The vertex set is
\[
    V = \{r, t, t'\} \cup \mathcal{S} \cup U,
\]
where $r$ is a root adjacent to all set nodes, and $t$ is a terminal 
adjacent to all element nodes, with a neighbor $t'$. The edge set is
\[
    E = (\{r\} \times \mathcal{S}) 
        \cup \{(S_i, e_j) \mid e_j \in S_i\} 
        \cup (U \times \{t\}) 
        \cup \{(r,t),\, (t, t')\}.
\]
That is, every set $S_i$ is adjacent to $r$, a set $S_i$ and element 
$e_j$ are adjacent iff $e_j \in S_i$, every element of $U$ is adjacent 
to $t$, and $t$ is adjacent to both $r$ and $t'$.

In the online adversarial setting, vertex $r$ is requested first, 
followed by the vertices in $\mathcal{S}$. When an element is requested 
in the online set cover instance, the corresponding vertex is requested 
in the incremental DS instance. Once the set cover instance is finished, the terminal
$t$ is requested, followed by $t'$, and then vertices in $U$ which were not yet 
requested. A visualization of the reduction appears in 
Figure~\ref{fig:osctoods}.

\begin{figure}[tp]
    \centering

\begin{tikzpicture}[
    scale=1,
    every node/.style={draw, circle, minimum size=6mm, inner sep=0pt, font=\small},
    >=stealth
]
\node (e1) at (-1,0) {$e_1$};
\node (e2) at (1,0)  {$e_2$};
\node (e3) at (3,0)  {$e_3$};
\node (e4) at (5,0)  {$e_4$};
\node (S1) at (0,-1.5)  {$S_1$};
\node (S2) at (2,-1.5)  {$S_2$};
\node (S3) at (4,-1.5)  {$S_3$};
\node (r) at (2,-3) {$r$};
\node (l)  at (-3, -1.5) {$t$};
\node (l2) at (-3,  0)   {$t'$};
\draw (r) -- (S1);
\draw (r) -- (S2);
\draw (r) -- (S3);
\draw (l) -- (r);
\draw (l2) -- (l);
\draw (S1) -- (e1);
\draw (S1) -- (e2);
\draw (S2) -- (e1);
\draw (S2) -- (e3);
\draw (S2) -- (e4);
\draw (S3) -- (e2);
\draw (S3) -- (e4);
\draw (l) -- (e1);
\draw (l) -- (e2);
\draw (l) -- (e3);
\draw (l) -- (e4);
\end{tikzpicture}
    \caption{The resulting reduction from a set cover instance with 4 elements $E = \{e_1,e_2,e_3,e_4\}$ and 3 sets $S_1 = \{e_1,e_2\}$, $S_2 = \{e_1,e_3,e_4\}$ and $S_3 = \{e_2, e_4\}$. The vertices $r$ and $S_1, S_2, S_3$ are requested first, and then when an element $e_i$ is requested in the online set cover instance its corresponding vertex is requested in the reduction. After the online set cover instance is finished, $t$ and $t'$ are requested, followed by the rest of the vertices in $E$.}
    \label{fig:osctoods}
\end{figure}
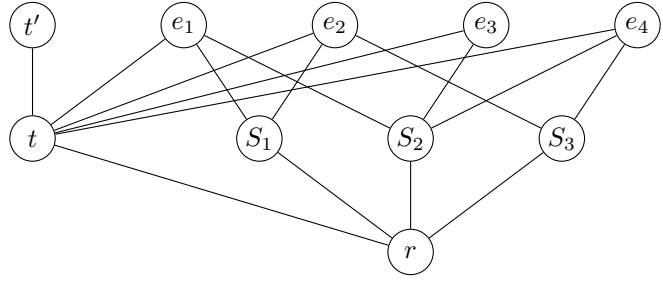

We start by proving the following lemma,
\begin{lemma}\label{lem:opt}
    There is an optimal solution for the reduced incremental instance that contains only elements belonging to $\{r,t\}\cup \mathcal{S}$.
\end{lemma}
\begin{proof}
    Let $\mathcal{I}$ be a solution to the reduced incremental instance. Vertex $r\in \mathcal{I}$, as $r$ is the first requested vertex. Let $\mathcal{S}' = \mathcal{S}\cap \mathcal{I}$ and $U' = U\cap \mathcal{I}$. For each $u\in U'$ we choose some set $s_u\in \mathcal{S}$ that contains it. Define $\mathcal{S}_{U'} = \{s_u | u\in U'\}$, and let $\mathcal{I'} = \{r,t\}\cup \mathcal{S}'\cup \mathcal{S}_{U'}$. 
    
    The solution $\mathcal{I}'$ dominates all requested vertices: each 
    $u \in U'$ is adjacent to $s_u \in \mathcal{I}'$, every set vertex 
    is dominated by $r$, and $r$ dominates itself. Since $t'$ has only 
    one neighbor $t$, any solution must contain $t$ or $t'$, so replacing 
    whichever appears in $\mathcal{I}$ with $t$ does not increase the 
    solution size. Finally, $|\mathcal{S}_{U'}| \le |U'|$ gives 
    $|\mathcal{I}'| \le |\mathcal{I}|$.\end{proof}

We now proceed to proving that the reduction satisfies several properties.
\begin{lemma}\label{lem6.1}
    Our reduction satisfies the following properties:
\begin{enumerate}
    \item The number of vertices in the resulting incremental instance is $n + m + 3$.\label{lem6.1.1}
    \item The optimal solution size in the incremental instance differs from the optimal set cover size by 2.\label{lem6.1.2}
    \item The reduction is online, meaning that each request in the online set cover instance is translated into a corresponding vertex in the online incremental DS instance.\label{lem6.1.3}
    \item There is an optimal solution that is connected.\label{lem6.1.4}
\end{enumerate}
\end{lemma}
\begin{proof}
Property~\ref{lem6.1.1} is true by the way the reduction is defined.

    The set cover solution selects a subset $\mathcal{S}'\subseteq \mathcal{S}$ to cover all the requested elements.
    The set $\{r, t\}\cup \mathcal{S}'$ is a valid solution to the reduction, since all vertices corresponding to sets in $\mathcal{S}$ are dominated by $r$, each requested element in $U$ must be covered by an element in $\mathcal{S}'$, and all other elements in $U$ are dominated by $t$.
 As a result, for each solution to the online set cover problem of size $k$ there is a solution to the incremental problem with size $k+2$.

    By Lemma~\ref{lem:opt}, let $\mathcal{I} \subseteq \{r,t\}\cup \mathcal{S}$ be the optimal solution to the reduced incremental instance, then $\mathcal{I}\backslash\{r,t\}$ is a feasible cover for the set cover instance of size $|\mathcal{I}| - 2$. Combining both directions proves Property~\ref{lem6.1.2}.

By Lemma~\ref{lem:opt} any optimal solution 
belongs to $\{r,t\} \cup \mathcal{S}$. The vertices $r$ and $\mathcal{S}$ 
are revealed at the start, before any element is requested. Subsequently, 
each element request in the online set cover instance corresponds exactly 
to the arrival of the matching vertex in the incremental DS instance. 
Once the set cover instance finishes, $t$ is requested, after which any 
remaining elements of $U$ are dominated by $t$. Thus the online algorithm does not need to
select any new vertices to dominate the remaining vertices requested in $U$, proving Property~\ref{lem6.1.3}.

 By Lemma~\ref{lem:opt}, there exists an optimal solution containing only $r$, $t$, 
and elements from $\mathcal{S}$, where each vertex in $\mathcal{S}$ is adjacent to $r$. Since $r$ 
is requested first and all other vertices in the optimal solution are connected to $r$, the 
partial solution remains connected at every step, thus proving Property~\ref{lem6.1.4}.
\end{proof}
We now prove the polynomial lower bound.
\newtheorem*{manualThm15}{Theorem 1.5}
\begin{manualThm15}
There is no polynomial-time online algorithm for \textnormal{incremental DS} or \textnormal{incremental CDS} with competitive ratio $o(\log^2 \Delta)$ unless \textnormal{NP} $\subseteq$ \textnormal{BPP}, even when the neighborhood of requested vertices is known.
\end{manualThm15}
\begin{proof}
    By the results of Korman et al.~\cite{lowerbound}, there are constants $a, b > 0$ such that there are online set cover
instances with $\hat{n}$ elements and $O(\hat{n}^a)$ sets such that the minimum number of sets required
to form a feasible cover at the end of the algorithm is $K = O(\hat{n}^b)$ and every polynomial randomized online
algorithm is $\Omega(\log^2 \hat{n})$-competitive on these instances unless $NP \subseteq BPP$. Moreover, this lower bound holds even when the full 
    set system is known in advance, as the hardness stems from the adversary's 
    choice of element requests rather than the structure of the sets. Since 
    the graph in our reduction is fully determined by the set system, knowing 
    the graph in advance is equivalent to knowing the set system, and thus any 
    lower bound transfers directly to the known-neighborhood setting of 
    incremental DS.

By Lemma~\ref{lem6.1}(\ref{lem6.1.2}), the optimal solution size remains $\Theta(K)$ in the reduced incremental instance.
Moreover, by Lemma~\ref{lem6.1}(\ref{lem6.1.3}), every polynomial algorithm is $\Omega(\log^2 \hat{n})$-competitive on the reduced incremental instance.

By Lemma~\ref{lem6.1}(\ref{lem6.1.1}), the reduced instance has $n = \mathrm{poly}(\hat{n})$ vertices.
In this construction, the maximum degree $\Delta$ is at least the degree of vertex $r$, which is $m = O(\hat{n}^a)$, so $\Delta = \Omega(\hat{n}^a)$. Since $\log \hat{n} = \Theta(\log \Delta)$, the lower bound of $\Omega(\log^2 \hat{n})$ translates to $\Omega(\log^2 \Delta)$.

Lemma~\ref{lem6.1}(\ref{lem6.1.4}) asserts that the optimal solution is connected, giving the same hardness result for incremental CDS.
\end{proof}
\newtheorem*{manualThm16}{Theorem 1.6}
\begin{manualThm16}
There is no deterministic online algorithm for \textnormal{incremental DS} or \textnormal{incremental CDS} with competitive ratio $o(\log^2 \Delta/\log \log \Delta)$, even when the neighborhood of requested vertices is known.
\end{manualThm16}
\begin{proof}The proof follows the same reduction as Theorem~\ref{thm1.7}, with the hardness bound replaced by the deterministic lower bound of Alon et al.~\cite{osc}.

    By the results of Alon et al.~\cite{osc}, for each $\hat{n}, \hat{m}$ with $\hat{n} = \Theta(\hat{m})$ there is an instance of online set cover with $\hat{n}$ elements and $\hat{m}$ sets such that the size of the optimal solution is 1 and there is no deterministic algorithm with competitive ratio $o(\log \hat{n}\log \hat{m}/\log\log (\hat{n} + \hat{m}))$ which simplifies to $o(\log^2 \hat{n}/\log\log \hat{n})$. As in the proof of Theorem~\ref{thm1.7}, since the graph is fully 
determined by the set system, the lower bound holds in the 
known-neighborhood model as well.

By Lemma~\ref{lem6.1}(\ref{lem6.1.2}) the optimal solution size is constant in the reduced incremental instance, and by Lemma~\ref{lem6.1}(\ref{lem6.1.3}) every deterministic algorithm is $\Omega(\log^2 \hat{n}/\log\log \hat{n})$-competitive on the reduced incremental instance.

Due to Lemma~\ref{lem6.1}(\ref{lem6.1.1}) the number of vertices in the reduced incremental instance is polynomial in $\hat{n}$. Additionally, since $\hat{n} = \Theta(\hat{m})$ in these instances and vertex $r$ has degree $\hat{m}$, the maximum degree satisfies $\Delta = \Omega(\hat{n})$. Therefore $\log \hat{n} = \Theta(\log \Delta)$, so the lower bound of $\Omega(\log^2 \hat{n}/\log\log \hat{n})$ translates to $\Omega(\log^2 \Delta/\log\log \Delta)$, proving the theorem for incremental DS.

Lemma~\ref{lem6.1}(\ref{lem6.1.4}) asserts that the optimal solution is connected, giving the same hardness result for incremental CDS.
\end{proof}
\subsection{Random Order Lower Bound}
We present here a lower bound of $\Omega(\log n)$ on random order incremental DS. This proves Theorem~\ref{thm1.10}. As a result, we can conclude that no algorithm has an $o(\log n)$ competitive ratio in the adversarial case, proving Theorem~\ref{thm1.8}.

We start with the result of Gupta et al.~\cite{random}.
\begin{theorem}[Theorem 5.1 in~\cite{random}]
\label{thmrand}
There is an instance of set cover with $n$ elements and $\Theta(n)$ sets, where any online algorithm for random order set cover is $\Omega(\log n)$-competitive. 
\end{theorem}

Let $(S, E)$ be the set cover instance used in Theorem~\ref{thmrand}, where $S$ is the collection of subsets and $E$ is the ground set of elements. We construct an incremental Dominating Set instance as follows:

The vertex set includes all elements in $E$. For each subset $s_i \in S$, we create a clique $S_i = {s_i^1, \dots, s_i^{n^3}}$. Each vertex in $S_i$ is connected to all elements $e \in s_i$. Additionally, the clique $R = {r^1, \dots, r^{n^5}}$ is added, where each $r^j \in R$ is connected to all vertices in every $S_i$.

Let $\eta$ be the event that a vertex in $R$ is requested first, and also for each $S_i$ a vertex $s\in S_i$ is requested before all the vertices in $E$. We show that event $\eta$ occurs with probability at least $1/2$.

\begin{lemma}\label{lem6.4}
    For sufficiently large $n$, $\P[\eta] \ge \frac12$.
\end{lemma}
\begin{proof}
    The size of $|R|$ is $n^5$ and the size of vertices which are not in $R$ is $\Theta(n) \cdot n^3 + n = \Theta(n^4)$. By definition of random order, the probability that a vertex in $R$ is not requested first is $\Theta\left(\frac{1}{n}\right) \le \frac{1}4$.

Next, for each $S_i$, the probability that all vertices in $E$ are requested after the first vertex of $S_i$ is at least $1 - \Theta\left(\frac{n}{n^3}\right) = 1 - \Theta\left(\frac{1}{n^2}\right)$. Applying the union bound over all $|S| = \Theta(n)$ sets, the probability that some $e \in E$ appears before any vertex in its corresponding $S_i$ is at most $\Theta(n \cdot \frac{1}{n^2}) = \Theta\left(\frac{1}{n}\right) \le \frac14$.

Applying the union bound to these two events, we get $\P[\eta] \ge \frac12$
\end{proof}

We now can prove the lower bound.

There is no online random-order algorithm for \textnormal{incremental DS} or \textnormal{incremental CDS} with competitive ratio $o(\log \Delta)$.
\newtheorem*{manualThm17}{Theorem 1.7}
\begin{manualThm17}
There is no online random-order algorithm for \textnormal{incremental DS} or \textnormal{incremental CDS} with competitive ratio $o(\log \Delta)$.
\end{manualThm17}
\begin{proof}
Assume the event $\eta$ occurs. Then some vertex in $R$ is selected first, and for every $s_i \in S$, a vertex $s_i^j \in S_i$, which is already dominated by $r$, is requested before any $e \in E$.

When vertices in $E$ are requested (in random order), each $e \in E$ is adjacent only to vertices $s_i^j$ such that $e \in s_i$. To maintain the domination invariant, the algorithm must select such an $s_i^j$ upon $e$'s arrival. Selecting $e$ itself would dominate only $e$, whereas selecting a neighboring $s_i^j$ may also dominate other future elements in $E$, so we may assume the algorithm avoids selecting $e$ directly.

This simulates the random order set cover problem of Theorem~\ref{thmrand}, where the algorithm must cover elements of $E$ using sets $s_i$. Thus, given $\eta$, any algorithm incurs competitive ratio $\Omega(\log n)$. Since by Lemma~\ref{lem6.4} the event $\eta$ occurs with constant probability, the expected competitive ratio is also $\Omega(\log n)$. As the reduction produces a graph of size polynomial in $n$, the lower bound applies to the random-order incremental Dominating Set problem. Furthermore, since the maximum degree $\Delta$ in this construction is $\Theta(n^5)$, this establishes a lower bound of $\Omega(\log \Delta)$.

To extend the result to the connected variant, note that the vertex $r \in R$ is adjacent to all $s_i^j$. Therefore, under $\eta$, both the optimal and the algorithm's solution are connected, so the same lower bound applies to the incremental Connected Dominating Set problem as well.
\end{proof}

Theorem~\ref{thm1.8} follows as a corollary of Theorem~\ref{thm1.10}, since any $o(\log \Delta)$-competitive algorithm for incremental DS would imply such an algorithm in the random-order setting, yielding a contradiction. The bound also holds when the neighborhood of requested vertices are known, since the reduction from online set cover already assumes that the set system is known in advance.

\subsection{Hardness of Approximation}\label{sec:hard}
We construct a reduction from the offline Dominating Set (DS) problem to the incremental Dominating Set and incremental Connected Dominating Set (CDS) problems. The reduction preserves the number of vertices and maintains the maximum degree up to a constant factor. This construction is used to prove Theorem~\ref{thm1.1}.

Let $G = (V, E)$ be an input graph with $V = \{v_1, \dots, v_n\}$, maximum degree $\Delta$, and minimum dominating set of size $\OPT$. We define a new graph $G' = (V', E')$ by adding $S = \{s_1, \dots, s_{n/\Delta}\}$ (wlog we assume that $\Delta$ divides $n$) and $U = \{u_1, \dots, u_n\}$ to $V$. The full vertex set is $V' = V \cup S \cup U$, consisting of $n(2 + 1/\Delta)$ vertices.

We define the following edge sets in the new graph:
\begin{align*}
    E_s &= \{(s_i, v_{\Delta \cdot i + j}) \mid i \in [n/\Delta],\ j \in [\Delta] \}, \\
    E_u &= \{(v_i, u_j) \mid (v_i, v_j) \in E \} \cup \{(v_i, u_i) \mid i \in [n] \}.
\end{align*}
The final edge set is $E' = E_s \cup E_u$.

Intuitively, the set $U$ is a duplicate of $V$ as the edges are mirrored from $V$ to $U$, and each vertex $v_i \in V$ is also connected to its counterpart $u_i \in U$. 
The vertices are requested in the following order: first the vertices in $S$, then those in $V$, and finally those in $U$.

The degree of each vertex $s_i$ is $\Delta$, the degree of $v_i$ is bounded by $2\Delta + 1$ and the degree of $u_i$ is at most $\Delta + 1$. Thus, the maximum degree is $O(\Delta)$. In Figure~\ref{fig:dstoids} an example of the reduction is given.

\begin{figure}[t]
    \centering
\begin{tikzpicture}[
    every node/.style={circle, draw, minimum size=8mm, font=\footnotesize, inner sep=0pt},
    >=stealth
]

\node (v1) at (-6,-1) {$v_1$};
\node (v2) at (-6, 1) {$v_2$};
\node (v3) at (-4,-1) {$v_3$};
\node (v4) at (-4, 1) {$v_4$};

\draw (v1)--(v2)--(v4)--(v3);
\draw (v2)--(v3);

\draw[thick, double,-implies, double distance=1.5mm] (-2.5,0)--(-1,0);

\node (u1) at (1,2) {$u_1$};
\node (u2) at (2.5,2) {$u_2$};
\node (u3) at (4,2) {$u_3$};
\node (u4) at (5.5,2) {$u_4$};

\node (v1p) at (1,0) {$v_1$};
\node (v2p) at (2.5,0) {$v_2$};
\node (v3p) at (4,0) {$v_3$};
\node (v4p) at (5.5,0) {$v_4$};

\node (s1) at (1.75,-2) {$s_1$};
\node (s2) at (4.75,-2) {$s_2$};
\node[dashed] (s3) at (3.25,-2) {$s_1'$};

\draw (s1)--(v1p);
\draw (s1)--(v2p);
\draw (s2)--(v3p);
\draw (s2)--(v4p);

\draw (v1p)--(u1);
\draw (v2p)--(u2);
\draw (v3p)--(u3);
\draw (v4p)--(u4);

\draw (v1p)--(u2);
\draw (v2p)--(u1);
\draw (v2p)--(u3);
\draw (v2p)--(u4);
\draw (v3p)--(u2);
\draw (v3p)--(u4);
\draw (v4p)--(u2);
\draw (v4p)--(u3);
\draw[dashed] (s1)--(s3);
\draw[dashed] (s2)--(s3);


\end{tikzpicture}
    \caption{An example of the reduction presented in Section~\ref{sec:hard}. 
    The original graph is shown on the left, and the reduced instance on the right. 
    Vertices in the reduced instance are requested from the bottom row to the top. 
    The vertex $s_1'$ (dashed) is added when constructing an incremental CDS instance.}
    \label{fig:dstoids}
\end{figure}
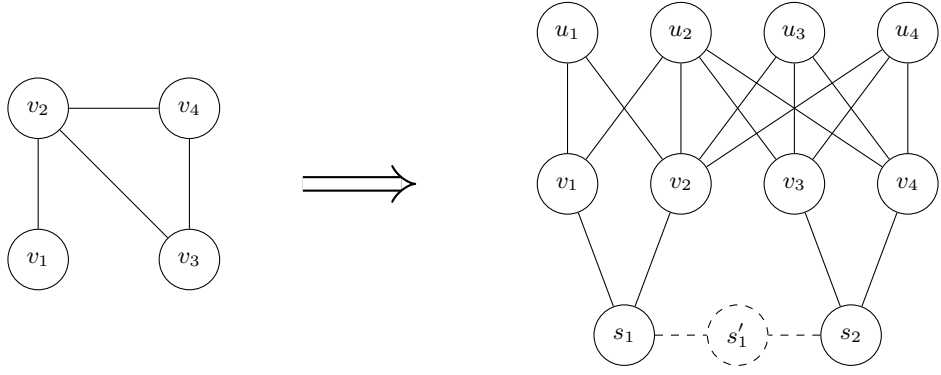

We now prove that the size of any solution is almost preserved between the two models
\begin{lemma}\label{lem6.3.3}
    If $D$ is a dominating set of $G$, then $D' = S \cup D$ is an incremental dominating set of $G'$ with $|D'| \le 3|D|$.
\end{lemma}
\begin{proof}
    If $D$ is a dominating set of $G$, then the set $D' = S \cup D$ is an incremental dominating set in $G'$. Each vertex $s_i$ is added to the solution immediately upon arrival. By construction, every vertex $v_i \in V$ is adjacent to at least one such $s_j$, ensuring that $V$ is dominated once it arrives. For each vertex $u_i \in U$, if $v_i \in D$, then $u_i$ is directly adjacent to $v_i \in D'$. Otherwise, since $D$ is a dominating set in $G$, there exists a neighbor $v_j \in D$ of $v_i$, and in this case $u_i$ is adjacent to $v_j \in D'$ via the edge $(v_j, u_i) \in E_u$.

The size of $D'$ is
\[|D| + \frac{n}\Delta \le |D| + 2\frac{n}{\Delta + 1} \le 3|D|.\]
The first inequality holds because $\Delta > 0$, and the second follows from the standard lower bound on the size of any dominating set in a graph with maximum degree $\Delta$.
\end{proof}
For the converse direction:
\begin{lemma}\label{lem6.3.4}
    If $I$ is an incremental dominating set of $G'$, then the set $D = \{v_i \mid v_i \in I \text{ or } u_i \in I\}$ is a dominating set of $G$ and satisfies $|D| = O(|I|)$.
\end{lemma}
\begin{proof}
Fix any $v_i \in V$.
\begin{itemize}
    \item If $v_i \in D$, then $v_i$ is trivially dominated.
    \item If $v_i \notin D$, but $v_i$ is adjacent in $G$ to some $v_j \in I$, then $v_i$ is dominated via $v_j \in D$.
    \item Otherwise, $u_i \in V'$ must be dominated in $G'$, and since $u_i \notin I$, there must exist a vertex $v_j \in I$ such that $(v_j, u_i) \in E_u$. By construction of $E_u$, this implies $(v_i, v_j) \in E$, and $v_j \in D$.
\end{itemize}

In all cases, $v_i$ is adjacent to some $v_j \in D$, so $D$ is a dominating set of $G$.

To bound the size of $D$, note that $S \subseteq I$, and each $v_i \in D$ corresponds to at most two vertices in $I$. Therefore,
\[
2|D| \ge |I| - |S| = |I| - \frac{n}{\Delta} \ge |I| - 2\frac{n}{\Delta + 1}.
\]
Using $|D| \ge \frac{n}{\Delta + 1}$,
\[
4|D| \ge 2|D| + 2\frac{n}{\Delta + 1} \ge |I|.
\]
\end{proof}

To extend the reduction to incremental CDS, we introduce an additional set of vertices $S' = \{s'_1, \dots, s'_{n/\Delta - 1}\}$ and add the edges $(s_i, s'_i)$ and $(s_i, s'_{i+1})$ for all $i \in [n/\Delta - 1]$. These vertices are requested after $S$ and must be selected to maintain connectivity. As a result, any valid incremental CDS solution must include the entire set $S'$ to maintain connectivity. Moreover, the reduction preserves the asymptotic relationship between the size of the original dominating set and the size of the corresponding incremental CDS solution.

We now prove the hardness of approximation.

\newtheorem*{manualThm19}{Theorem 1.9}
\begin{manualThm19}
\textnormal{Incremental DS} and \textnormal{incremental CDS} are \textnormal{NP-complete}. Furthermore, unless \textnormal{NP} $\subseteq$ \textnormal{BPP} there is no polynomial-time $o(\log\Delta)$ approximation algorithm for either problem.
\end{manualThm19}
\begin{proof}
The reduction preserves the maximum degree up to a constant factor. By Lemmas~\ref{lem6.3.3} and~\ref{lem6.3.4} there is a linear relationship between the sizes of optimal solutions in $G$ and $G'$. Thus, the existence of a polynomial-time $o(\log \Delta)$-approximation algorithm for incremental Dominating Set would imply the same for the classical Dominating Set problem, contradicting known hardness results~\cite{hardness}.
\end{proof}
\section{Future Directions}
Our results provide a foundation for understanding the incremental dominating set problem and its variants. Future work could explore several additional directions.

A natural next step is to improve the competitive ratio for the weighted case in the known-neighborhood setting (Theorem~\ref{thm1.11}), from the current bound of $O(\log n \log \Delta)$ to $O(\log^2 \Delta)$. While the techniques that improve the bound in the unweighted case do not directly extend to the general case, we believe that a more refined algorithm achieving this bound exists. Another promising direction in this setting is to extend our results for the incremental CDS problem to the weighted case.

Beyond worst-case analysis, it would be interesting to study random-order models, where the input is revealed in a randomized rather than adversarial order, potentially yielding improved competitive ratios. Unlike in other online settings, in the incremental model the input order affects the set of feasible solutions and the value of the optimum. Even the definition of competitive ratio in this context becomes nontrivial, as the expressions $\E_{\pi}[\mathcal{A}(\pi)/\textnormal{OPT}(\pi)]$ and $\E_{\pi}[\mathcal{A}(\pi)] / \E_{\pi}[\textnormal{OPT}(\pi)]$ may differ.
\bibliographystyle{splncs04}
\bibliography{bibfile}

\input{CDStoLP}

\input{omit}

\end{document}

%% file: CDStoLP.tex
\appendix

\section{Relaxation for Incremental Connected Dominating Set}
\label{sec:cdstolp}

In this appendix, we provide a detailed proof that the LP formulation presented in Section~\ref{sec:CDSascover} is a valid relaxation for the incremental CDS problem. Recall that we classify each arriving vertex $v_i$ into three types based on how it interacts with existing components:

\begin{enumerate}
    \item \textbf{Opener}: not connected to any existing component ($N[v_i] \cap V_i = \{v_i\}$)
    \item \textbf{Attacher}: connected to exactly one existing component
    \item \textbf{Joiner}: connects two or more existing components
\end{enumerate}

We prove that the constraints defined for each type in Section~\ref{sec:CDSascover} form a valid relaxation by showing both directions: any incremental CDS induces a feasible LP solution, and any integral feasible LP solution induces an incremental CDS.

Denote by $C_1^i,\dots,C_{h_i}^i$ the connected components induced by the
vertices requested up to step~$i$.

\paragraph{Forward direction: incremental CDS implies LP feasibility.}
We show that any feasible incremental CDS $D$ induces a feasible
integral solution to the LP.

\paragraph{Case 1: Opener.} 
By definition $v_i\in D$, hence $x_i=1$.

\paragraph{Case 2: Attacher.} 
Let $C^{i-1}$ denote the unique component adjacent to $v_i$. 
The LP constraint is $\sum_{v_j\in N(v_i)\cap V_i} x_j \ge 1$.
\begin{itemize}
\item If $v_i\notin D$, then $v_i$ is dominated by some $v_j \in N(v_i) \cap D$, 
yielding $\sum_{v_j\in N(v_i)\cap V_i} x_j \ge 1$.

\item If $v_i\in D$, connectivity requires a neighbor
$v_j\in N(v_i)\cap C^{i-1}$ with $v_j\in D$, 
satisfying the constraint.
\end{itemize}

\paragraph{Case 3: Joiner.} 
Let $v_i$ connect components $C_1^{i-1},\dots,C_{k_i}^{i-1}$ with $k_i\ge2$.
Since these components are disjoint, any path between them passes through $v_i$.
Connectivity of $D$ requires $v_i\in D$, thus $x_i=1$.
Moreover, for each $h$, connectivity requires 
$v_j\in N(v_i)\cap C_h^{i-1}$ with $v_j\in D$, yielding
\[
\sum_{v_j\in N(v_i)\cap C_h^{i-1}} x_j \ge 1 .
\]

\paragraph{Reverse direction: LP feasibility implies incremental CDS.}
We prove by induction on~$i$ that any integral feasible LP solution induces an incremental CDS.

\paragraph{Case 1: Opener.} 
The LP gives $x_i=1$, hence $v_i\in D$. 
The component contains only $v_i$, which is trivially connected.

\paragraph{Case 2: Attacher.} 
Let $C^{i-1}$ denote the unique component adjacent to $v_i$.
The constraint $\sum_{v_j\in N(v_i)\cap V_i} x_j \ge 1$ ensures 
$v_j\in N(v_i) \cap V_{i-1}$ with $v_j\in D$, so $v_i$ is dominated.
\begin{itemize}
\item If $v_i\notin D$, then $D\cap C^{i}$ trivially remains connected.

\item If $v_i\in D$, then some $v_j\in N(v_i)\cap C^{i-1}$ satisfies $v_j\in D$.
By induction, $D\cap C^{i-1}$ is connected, so $D\cap C^{i}$ is connected.
\end{itemize}

\paragraph{Case 3: Joiner.} 
The LP gives $x_i=1$, hence $v_i\in D$.
For each component $C_h^{i-1}$ joined by $v_i$, the constraint 
$\sum_{v_j\in N(v_i)\cap C_h^{i-1}} x_j \ge 1$
guarantees $v_j \in D\cap C_h^{i-1}$ adjacent to $v_i$.
By induction, each $D\cap C_h^{i-1}$ is connected. Since $v_i$ has a neighbor in each $D\cap C_h^{i-1}$, adding $v_i$ connects these sets into a single connected component.

\medskip
\noindent
Therefore, the LP is a valid relaxation of the incremental CDS problem.

%% file: omit.tex
\section{Omitted Proofs}
\label{sec:omit}
\subsection{Proof of Lemma~\ref{lem:subset}}
In this part we prove lemma~\ref{lem:subset} which was omited from the main text
\newtheorem{manualLemma}{Lemma} 
\renewcommand{\themanualLemma}{4.1} 
\begin{manualLemma}
If $v_i$ is uncovered at step $i$, then there exists a subset of 
$N[v_i]\cap V_i$ of size at most $6\log m$ such that adding the vertices 
in the subset to $D$ ensures that the potential does not increase, i.e.,
\[
\Phi_i \le \Phi_s,
\]
where $\Phi_s$ is the potential prior to the weight increases in Line~\ref{alg2l7}.
\end{manualLemma}

Recall that $\Phi_s$ is the value of the potential before the weight increase, and $\Phi_i$ is its current value. When a new uncovered vertex is revealed, the weights of vertices covering it are doubled until it is covered in the fractional solution. 

For a vertex $v_j$, let $c'_{v_j}$ denote its original cover, $\delta_j$ the increase in the fractional variable $x_j$, and let
\[
\delta_j' = \sum_{v_k\in N[v_j]\cap V_j}\delta_k = c_{v_j} - c'_{v_j}
\]
be the total increase in the value of the fractional covering of $v_j$.

We prove the lemma by describing a randomized procedure with two properties:
\begin{enumerate}
    \item It never selects more than $6\log m$ vertices.
    \item It satisfies $\E[\Phi_i] \le \Phi_s$.
\end{enumerate}
These two properties together imply the existence of a subset of size at most $6\log m$ with $\Phi_i \le \Phi_s$.

\paragraph{The procedure.}
Repeat the following operation $6\log m$ times: with probability $\delta_j/2$, select each vertex $v_j$ (at most one vertex is chosen per iteration). Since the total weight of vertices covering the requested element is at most 2, we can ensure that at most one vertex is selected in each iteration.

By construction, the total number of selected vertices is at most $6\log m$, proving property 1.

To prove property 2, fix a vertex $v_j \in U_{i-1}$. The probability that $v_j$ is not covered after the procedure is
\[
\left(1 - \frac{\delta'_j}{2}\right)^{6\log m} \;\le\; m^{-3\delta'_j}.
\]
Thus the expected contribution of $v_j$ to the potential is
\[
(1 - m^{-3\delta'_j}) \cdot 0
+ m^{-3\delta'_j}\cdot m^{3(c'_{v_j} + \delta'_j)}
= m^{3c'_{v_j}}.
\]
This is equal to the original contribution of $v_j$ to the potential. Since this holds for all $v_j \in U_{i-1}$, the expected potential does not increase, proving the second property.
\subsection{Proof of Lemma~\ref{lem5.4}}
In this part, we prove lemma~\ref{lem5.4} about the change of the potential function throughout the run of algorithm~\ref{alg:weighted}.

\renewcommand{\themanualLemma}{4.7} 
\begin{manualLemma}
Denote $\Phi_i^j$ as the value of the potential function after the $j$th iteration of the loop in Line~\ref{alg3l5}, and $\Phi_i^0$ as the potential at the beginning of the loop. Then,
\[\forall j: \Phi_i^j \le \Phi_i^{j-1}\]
\end{manualLemma}
To prove the lemma, we start by denoting $\delta_j$ as the change in the parameter $x_j$ of vertex $v_j$. We prove the following helper lemma.
\begin{lemma}\label{lemb.2}
    If we add $v_j$ to $D$ with probability $\left(1-m^{-2\delta_j}\right)$ then 
    \[\E[\Phi_i^j]\le \Phi_i^{j-1}.\]
\end{lemma}
Thus, if the potential increases when changing $x_j$, then adding $v_j$ to $D$ ensures we must get $\Phi_i^j \le \Phi_i^{j-1}$, proving lemma~\ref{lem5.4}.

\begin{proof}[Proof of Lemma~\ref{lemb.2}]
We begin with the first part of the potential, $\sum_{u \in U_i} m^{2c_u}$.
For any $u \in U_i$ not adjacent to $v_j$, $c_u$ does not change, and the corresponding term $m^{2c_u}$ remains unchanged.
If $u$ is adjacent to $v_j$, then the expected contribution is
\[
(1-m^{-2\delta_j})\cdot 0 + m^{-2\delta_j} \cdot m^{2(c_u + \delta_j)} = m^{2c_u},
\]
so the expected value of each term $m^{2c_u}$ remains equal to its previous value.

    We now bound the second term of the potential function. Let
\[
T = \exp\!\left(\frac{1}{2\alpha}\Big(\sum_{v\in D}w(v) - 3\log m \sum_{v_j \in V_i} w(v_j)x_j\Big)\right)
\]
denote the previous value of the second term, and let $T'$ denote its value after the update. Then
\[
T' = T \cdot \exp\!\left(-\tfrac{3\log m}{2\alpha}\,w(v_j)\delta_j\right)
     \cdot \E\!\left[\exp\!\Big(\tfrac{1}{2\alpha}w(v_j)\cdot \mathds{1}[v_j\in D]\Big)\right],
\]
where $\mathds{1}[v_j\in D]$ is the indicator variable for adding $v_j$.

    The expected value is
    \setcounter{equation}{0}
    \begin{align}
    &\E\left[\exp\left(\frac1{2\alpha}w(v_j)\mathds{1}[v_j\in D]\right)\right]\nonumber\\
        &= m^{-2\delta_j} + (1-m^{-2\delta_j})\cdot \exp\left(\frac{w(v_j)}{2\alpha}\right)\\
        &\le 1+2\delta_j\log m\left(\exp\left(\frac{w(v_j)}{2\alpha}\right) - 1\right)
        \label{aln2}\\
        &\le 1 + 2\delta_j\log m\frac{3w(v_j)}{4\alpha}\label{aln3}
        \\&\le \exp\left(\frac{3\log m\cdot w(v_j)\delta_j}{2\alpha}\right).\label{aln4}
    \end{align}
    Inequality~(\ref{aln2}) follows since $e^{-x} + (1-e^{-x})\cdot z \le 1 + x(z-1)$. 
Inequality~(\ref{aln3}) uses $e^x - 1 \le \tfrac{3}{2}x$ for $x \in [0,1/2]$, and note that $w(v_j)/(2\alpha) \le 1/2$, since heavy vertices with $w(v_j)>\alpha$ never appear in the fractional solution. 
Finally~(\ref{aln4}) follows from $1 + x \le e^x$ for $x \ge 0$.

We conclude that $\E[T'] \le T$. Since the first part of the potential does not change, and the second part does not increase, by linearity of expectation
\[
\E[\Phi_i^j] \le \Phi_i^{j-1}.\qedhere
\]
\end{proof}